\documentclass{llncs}
\usepackage[T1]{fontenc}

\usepackage{mathtools}
\usepackage{cases}

\usepackage{amsmath}

\usepackage[T1]{fontenc}

\usepackage{amssymb}
\usepackage{graphicx}

\usepackage{graphics}
\usepackage{epsfig}
\usepackage{latexsym}

\usepackage[]{algorithm2e}

\usepackage{epic}
\usepackage{curvesls}
\usepackage[T1]{fontenc}

\usepackage{algpseudocode}

\usepackage{cases}

\usepackage{tikz} 
\usetikzlibrary{arrows,decorations.pathmorphing,positioning,fit,trees,shapes}

\newcommand{\var}{\ensuremath{\mathsf{var}}}
\newcommand{\lit}{\ensuremath{\mathsf{Lit}}}
\newcommand{\B}{\ensuremath{\mathsf{B}}}
\newcommand{\cls}{\ensuremath{\mathsf{Cls}}}
\newcommand{\cnf}{\ensuremath{\mathsf{Cnf}}}

\newcommand{\clsset}{{\ensuremath{\cal{C}}}}

\newcommand{\fa}{\ensuremath{\mathsf{false}}}
\newcommand{\tr}{\ensuremath{\mathsf{true}}}

\newcommand{\F}{\ensuremath{\mathsf{f}}}
\newcommand{\T}{\ensuremath{\mathsf{t}}}

\newcommand{\node}{\ensuremath{\mathsf{Node}}}

\newcommand{\ph}{\ensuremath{\mathsf{P}}^h}
\newcommand{\pl}{\ensuremath{\mathsf{P}}^l}

\newcommand{\phc}{\ensuremath{\mathsf{\overline{P}}}^h}
\newcommand{\plc}{\ensuremath{\mathsf{\overline{P}}}^l}

\newcommand{\apply}{\ensuremath{\mathsf{Apply}}}
\newcommand{\ax}{\ensuremath{\mathsf{Axiom}}}
\newcommand{\jn}{\ensuremath{\mathsf{Join}}}

\newcommand{\pr}{\ensuremath{\mathsf{Projection}}}
\newcommand{\wk}{\ensuremath{\mathsf{Weakening}}}

\newcommand{\low}{\ensuremath{\mathsf{low}}}
\newcommand{\high}{\ensuremath{\mathsf{high}}}
\newcommand{\paths}{\ensuremath{\mathsf{Path}}}

\newcommand{\fc}{\ensuremath{\mathsf{\tau}}}

\newcommand{\pathsf}{\ensuremath{\mathsf{Path^f}}}

\newcommand{\rt}{\ensuremath{\mathsf{root}}}

\newcommand{\isomBDD}{\equiv}

\newcommand{\tvo}{\prec}

\newcommand{\res}{\ensuremath{\mathsf{res}}}

\newcommand{\respairs}{\ensuremath{\mathsf{Res}}}

\newcommand{\dec}{{\small{\ensuremath{\mathsf{Dec}}}}}

\newcommand{\newnode}{{\small{\ensuremath{\mathsf{Node}}}}}

\newcommand{\asg}{\ensuremath{\mathsf{A}}}

\newcommand{\af}{\ensuremath{\mathsf{F}}}

\newcommand{\rep}{\gamma}

\newcommand{\asset}{\mathcal{A}}

\newcommand{\less}{\unlhd}

\newcommand{\Tr}{\ensuremath{\mathsf{t}}}
\newcommand{\Fa}{\ensuremath{\mathsf{f}}}

\newcommand{\php}{\ensuremath{\mathsf{PHP}}}
\newcommand{\gphp}{\ensuremath{\mathsf{\overline{PHP}}}}

\begin{document}

%\title{Resolution  Simulates Polynomially \\  Ordered Binary Decision Diagrams for\\  Formulas in  Conjunctive Normal Form}
%\author{Olga Tveretina\\~\\
%School of Computer Science\\ University of Hertfordshire\\ United Kingdom\\ ~ \\
%o.tveretina@herts.ac.uk
%}
%%\institute{
%School of Computer Science\\ University of Hertfordshire\\ United Kingdom\\
%\email{ o.tveretina@herts.ac.uk}}

%\title{Resolution  Simulates OBDDs  Polynomially  for\\  Formulas in   Conjunctive Normal Forms}

\title{Resolution  Simulates Ordered Binary Decision Diagrams  \\ for  Formulas in  Conjunctive Normal Form}

\author{Olga Tveretina}
\institute{School of Computer Science\\
    University of Hertfordshire\\ United Kingdom\\
    \email{o.tveretina@herts.ac.uk}}

\date{}

\maketitle

\begin{abstract}
A classical question of propositional logic is one of the shortest proof of a tautology.  A related fundamental problem is to determine the relative efficiency of standard proof systems, where the relative complexity is  measured using the notion of polynomial simulation.

Presently, the state-of-the-art satisfiability algorithms are based on resolution in combination with search.
An Ordered Binary Decision Diagram (OBDD) is a data structure that is used to represent Boolean functions.

 Groote and Zantema have  proved that there is  exponential separation between resolution and a   proof system based on limited OBDD derivations. 
 However, formal comparison of these methods is not straightforward because OBDDs work on arbitrary formulas, whereas resolution can only be applied to formulas in Conjunctive Normal Form (CNFs). 

 Contrary to popular belief, we argue that resolution simulates OBDDs polynomially if we limit both to CNFs and thus answer negatively  the open question of Groote and Zantema   whether there exist unsatisfiable  CNFs  having polynomial OBDD
refutations and requiring exponentially long resolution refutations.

 \end{abstract}

\section{Introduction}

Propositional proof complexity is the study of the lengths of proofs of statements expressed as propositional formulas. It is tightly connected in many ways to computational complexity,  classical proof theory  and practical questions of  automated deduction.

A classical question of propositional proof complexity is one of the shortest proof of a tautology. A related fundamental problem is to determine the relative efficiency of standard proof systems, where the relative complexity is measured using the notion of polynomial simulation.

{\bf{Proof systems}.}
Propositional proof systems were defined by   Cook and Reckhow as polynomial-time functions which have as their range the set of all tautologies \cite{CR1979}.
They also noticed  that if there is no propositional proof system that admits proofs polynomial in size of the input formula then the complexity classes NP and co-NP are different, and hence P$\neq$ NP.

In \cite{AKV2004}, Atserias, Kolaitis and Vardi generalised the notion of a refutational propositional proof system,  viewing it as a special case of constraint propagation.
 Their proof system   consists  of  the following four rules: 
(1) $\ax$ defines the initial set of constraints;
(2) $\jn$ combines s  two constraints by intersecting   two relations and extending them   to all variables occurring in either one of them;
(3) $\pr$ computes the projection of a constraint which is the existential quantification;
(4) $\wk$ relaxes the constraint by enlarging its relation.

This generalisation  brings the methods of constraint propagation to the area of proof complexity. On the other hand it introduces new classes of proof systems.  The existing 
refutational proof systems can be viewed as a special case of this constraint propagation system.

{\bf{Efficiency of proof systems}.}
One of the most fundamental problems in the area of propositional
proof complexity is to determine the relative efficiency of standard
proof systems as it has been introduced  by Cook and Reckhow in \cite{CR1979} who found it useful to separate the idea of providing a proof from that being efficient.

  Proof systems are compared according to their strength using the notion of polynomial simulation.
A proof system $S_1$ 
simulates polynomially  a proof system $S_2$  if every tautology has
proofs in $S_1$ of size at most polynomially larger than in $S_2$. Proof systems $S_1$ and $S_2$ are equivalent if they simulate  each other polynomially.
%
% If it
%can be shown that  the proofs of some formulas in $Q$ are
%exponentially longer than those in $P$, we consider $P$ as a strictly
%better proof system than $Q$.

 Although substantial
progress has been made in determining the relative complexity of proof
systems and in proving strong lower bounds for some relatively weak
proof systems, some major problems still remain unsolved.

As  it is defined by Razborov in    \cite{R03}, the question of existence of an efficient proof has to be separated from another important question how to  find such a proof efficiently and whether this search adds substantially to the inherent complexity of finding the shortest  proof in a specific proof system. It is formalised by the notion of automatizability: a proof system is automatizabile if it  produces a proof of a tautology  in time polynomial in the size of its smallest proof   \cite{BPR00}.

 The current  proof systems of practical use are not  automatizabile (or  just weakly automatizabile). That is why in addition to lower bounds, there is practical interest in understanding relative efficiency of somewhat  weaker proof systems. Interesting  examples of such systems are those  based either on general resolution or on classical OBDDs not utilising existential quantification.

{\bf{Resolution versus  OBDDs}.}
In the automated reasoning community
resolution and OBDDs are  popular  techniques for solving the  propositional satisfiability problem abbreviated as SAT.  In fact, both resolution and OBDDs  are families of algorithms, where each corresponds to a specific way of making choices. 

 Resolution  underlies  the vast majority of all proof search
techniques in this area. For example, the DPLL algorithm 
  \cite{DLL1962}, as well as the  
clause learning methods are highly optimised 
implementations of resolution \cite{PD09}.  It has been shown in \cite{PD10} that modern SAT solvers
simulate resolution polynomially.

 An
OBDD  is a canonical data structure
that is used for the symbolic representation of Boolean
functions~\cite{B1986,W2000}. Atserias et al. introduced and studied a proof system operating with OBDDs as a special case of constraint propagation \cite{AKV2004}.  They proved that OBDD based refutations polynomially simulate  resolution if they utilise existential quantification. That is  OBDD based  proof systems containing all four rules $\ax$, $\jn$, $\wk$ and $\pr$  are strictly stronger than
resolution but they are still exponential \cite{K2008}.

In the following we consider the OBDD proof system which  contains just  two rules,    $\ax$ and $\jn$. These rules are equivalent to    the $\apply$ operator as it is defined in \cite{B1986}.

Benchmark studies show incomparable behaviour of resolution and such OBDD based systems  \cite{US1994}.   
Groote and Zantema proved that resolution and OBDDs do not simulate
each other polynomially on arbitrary inputs for
 limited OBDD derivations  \cite{GZ2003}.   Tveretina, Sinz and Zantema strengthened the above
result and presented a class of CNFs hard for an arbitrary OBDD
derivation and easy for resolution \cite{TSZ10}.

In general, formal
comparison of resolution and OBDDs  is not straightforward because the later work
on arbitrary formulas, whereas resolution can take as an input only
CNFs. 
We argue that resolution simulates OBDDs polynomially if we limit both to CNFs. Thus we answer negatively    the  open question of Groote and Zantema posed in  \cite{GZ2003} whether there exist unsatisfiable  CNFs  having polynomial OBDD
refutations and requiring exponentially long resolution refutations.

{\bf{Previous  work}.}
There are several works which study the relative efficiency  of resolution based and OBDD based proof systems. The most relevant studies  to our setting are the following ones.

 Peltier shows in  \cite{P2008}  that  resolution augmented with the extension rule polynomially simulates OBDDs in the following sense: for any unsatisfiable formula $\varphi$ there exists a refutation of $\varphi$  with the size polynomially bounded by the maximal size of the reduced OBDDs corresponding to the subformulas occurring in $\varphi$.
As mentioned before,  Atserias et al  prove in   \cite{AKV2004} that  OBDD based refutations utilising existential quantification polynomially simulate resolution; moreover they are exponentially stronger. 
Groote and Zantema construct in \cite{GZ2003}  biconditinal formulas that have short OBDD refutations and after transforming them into CNFs they require exponentially long resolution proofs. But the same formulas after transformation into CNFs have exponentially long OBDD proofs.

{\bf{Main result}.}
We show that for any unsatisfiable CNF$\varphi$ there exists a resolution refutation of $\varphi$  with  the size polynomially bounded by the  size of an OBDD based refutation of $\varphi$ if it consists of  two rules $\ax$ and $\jn$ and uses two standard reduction rules, elimination and merging.  We now formally state the theorem.

\begin{theorem}\label{theorem:main2}
 Assume an unsatisfiable CNF $\varphi$. If there is an OBDD refutation of $\varphi$ with two rules $\ax$ and $\jn$ of size $n$ then there is a resolution refutation of $\varphi$  of size 
 $O(n^2)$.
\end{theorem}

Our main argument is based on the  idea    that  the elimination rule can be simulated  by applying the resolution rule on the variable corresponding to the eliminated node  \cite{P2008}. But we use it differently.

We strengthen this idea and 
  prove that  the number of resolution steps corresponding to the elimination of a node  is  bounded by the number of clauses  in the input CNF $\varphi$. Moreover, we show that it is  an  invariant property: although    resolution steps generate new clauses,  the number of resolution steps needed to simulate   elimination of a node in an  intermediate OBDD remains  bounded by the number of clauses  of  $\varphi$ encoded by this OBDD.
Furthermore, we show that it is sufficient to simulate only the elimination rule, that is  the merging rule plays no role in the context.

{\bf{The remainder of the paper}.}
We give the necessary background in Section \ref{sec:preliminaries}. In Section \ref{sec:proof_systems}  we 
 introduce two proof systems of interest: one is based on resolution and the other corresponds to  OBDD derivations based on the Axion and Join rules.  
In Section \ref{sec:simulation} we show how to simulate the elimination rule using resolution   and in Section \ref{sec:main} we prove our main result.
Section \ref{sec:conclusion}  contains concluding remarks.

\section{Preliminaries}\label{sec:preliminaries}

\subsection{Propositional Logic and Conjunctive Normal Forms}

In this section we recall some basic notations about propositional logic and only provide a short overview of the main definitions. 
%A more detailed presentation  can be found in    \cite{Burris1997}.

In the following we consider propositional formulas in  Conjunctive Normal Form (CNF)  built using 
variables from a set $\var$. A
literal $l$ is either a variable $x$ or its negation $\lnot x$ 
with $\var(l)=x$. A clause $C$  is a
disjunction of literals, and a CNF $\varphi$  is a conjunction of clauses. By $\cnf$ we denote the set of all CNFs.

%Set notations for CNFs.

 We define $\cls(\varphi)$
to be the set of clauses, $\lit(\varphi)$ the set of literals, and
$\var(\varphi)$ the set of variables contained in the CNF $\varphi$.

 We use
$\varphi|_l$ to denote the CNF obtained from $\varphi$ by deleting all clauses  containing a literal $l$ and removing $\lnot l$ from the rest of the clauses. Note that  $\varphi=\varphi|_{l}$ if $l\not\in \lit(\varphi)$.

A truth assignment is a function $\asg:\var\rightarrow
\{\tr,\fa\}$.  We denote by $\asset$ the set of all possible assignments.
%, and by $\asset_n$ the set of all possible assignments of $n$ variables $x_1,x_2,\dots, x_n$. 
The truth values of literals, clauses and CNFs are defined in a standard way.

We write $\asg\models \varphi$ if $\varphi$ evaluates to
$\tr$ for the assignment $\asg$, otherwise  we write $\asg\not\models \varphi$.  
We say that $\varphi$ is {unsatisfiable} if  $\asg\not\models \varphi$ for any $\asg\in\asset$, otherwise it is satisfiable; $\varphi$ is a {tautology} if  $\asg\models \varphi$ for any $\asg\in\asset$.

We say that two CNFs $\varphi$ and $\psi$ are logically equivalent, denoted $\varphi\equiv \psi$,  if $\asg \models \varphi$ if and only if $\asg \models \psi$ for any $\asg\in \asset$.

We use 
$\top$ for the empty set of clauses  and $\bot$ for  the CNF consisting of the empty clause. By definition the empty clause is unsatisfiable that is it is equivalent to $\fa$,  and the the empty set of clauses is equivalent to $\tr$.

\subsection{Ordered Binary Decision Diagrams}

The concept of ordered binary decision diagrams (OBDDs)  was first proposed by Lee
in~\cite{L59}  as a means to represent propositional formulas  (Boolean functions) compactly as
directed acyclic graphs (DAGs). Then it was further developed to a data structure by
Acers \cite{A78} and Boute \cite{B76}, and subsequently by Bryant  
\cite{B1986}.

\begin{definition}[An OBDD]\normalfont{
An OBDD  $\B$ is a directed acyclic graph satisfying the following:
\begin{enumerate}
\item  it has a unique node  called the root and denoted by  $\rt(\B)$;
\item  each inner node $p$ is  labeled by the  propositional variable $\var(p)$ and has  exactly two successors, a $\fa$-successor and a $\tr$-successor; 
\item the inner nodes build the set $\node(\B)$, and the labels build the set $\var(\B)$;
\item  each  leaf  node is  labeled by either  $\tr$ or $\fa$;  
\item there is a total variable order $\tvo$ such that for each transition from the inner node with label $x$ to the inner node with label $y$ we have that $x\tvo y$.
\end{enumerate}
}
\end{definition}

We use $\high(p)$ and $\low(p)$  to denote  the OBDDs rooted at the $\tr$-successor  and  the $\fa$-successor of $p$; 
$|\B|$ is  the size of
 $\B$, that is the number of its inner nodes.   
 
We use $\B_1\isomBDD \B_2$ to denote that $\B_1$ and $\B_2$ are isomorphic OBDDs defined  as follows:
\begin{itemize}
\item both $\B_1$ and $\B_2$  consist either  of  the node $\tr$ or  of the  node $\fa$;

\item $\var(\rt(\B_1))=\var(\rt(\B_2))$, $\high(\B_1)\isomBDD \high(\B_2)$ and $\low(\B_1)\isomBDD \low(\B_2)$.
\end{itemize}

 OBDD operations  are applicable only
  to OBDDs that respect the same variable ordering.  
  To shorten the
  notations in the rest of the paper,  we assume without explicitly stating that all
  the variables agree on the common variable order $x_1\tvo x_2 \tvo x_3\tvo \dots$ when considering different OBDDs
  in the same context.

\begin{definition}[Path]\normalfont{
 A \emph{path} of an OBDD $\B$ is a sequence
 $\alpha=l_1\ldots l_k$  of literals with  $k\geq 1$ such that  there are  $p_1, \ldots, p_{k}\in\node(\B)$, where
 \begin{itemize}
 \item $p_1=\rt(\B)$;
 \item  for $1\leq i< k$, either $p_{i+1}=\rt(\high(p_i))$ and $l_i=\var(p_i)$,  or  $p_{i+1}=\rt(\low(p_i))$ and $l_i=\neg \var(p_i)$;
 \item $\high(p_k)\in\{\fa,\tr\}$ if $l_k =\var(p_k)$;
 \item $\low(p_k)\in\{\fa,\tr\}$ if $l_k =\lnot \var(p_k)$. 
  \end{itemize}
}
\end{definition}
  %We see a path of an OBDD as a sequence of literals.  

We use   
$\paths(\B)$  to denote the set of all paths of $\B$, and  $\pathsf(\B)$ to denote the set of all paths that go  to the $\fa$ node.
By  $\paths(p)$ we mean  the set of all paths that go through the inner node $p$, and  $\pathsf(p)=\paths(p) \cap \pathsf(\B)$.

A path can be seen  as a conjunction  of 
literals.  In this way, each  path $\alpha=l_1 \ldots l_k$ naturally induces the set $\asset(\alpha)$ of
truth assignments  evaluating each $l_i$ to
$\tr$. 
We write $\alpha\not\models C$ for a clause $C$ if $A\not\models C$ for any $A\in \asset(\alpha)$. 
Where it is convenient, we see a path as a set of literals and use the set notations.

We say that a CNF $\varphi$ and an OBDD $\B$ are logically equivalent if for every assignment  $\asg\in\asset$, $\alpha\not\models \varphi$ if and only if  there is a path $\alpha\in\pathsf(\B)$ such that $\asg\in \asset(\alpha)$.

For any CNF $\varphi$ and  OBDD $\B$,  we use   $\varphi\less \B$ to denote that $\varphi$ and $\B$ are logically equivalent and for each path $\alpha\in \pathsf(\B)$ there is a clause $C\in\cls(\varphi)$ such that $\alpha\not\models C$.

\subsection{OBDD Construction}

The straightforward  way to construct an OBDD  is to start with  a binary decision
tree and then incrementally eliminate redundancies and identify identical
subtrees.
The other more efficient way
 follows the structure of the propositional formula. Such algorithms start with building
OBDDs for variables or literals, and then construct more complex OBDDs by using OBDD operations for logical connectives. 

Algorithm \ref{algo:and}  presented  below  (from   \cite{SB2006}) takes as an input   two OBDDs $\B_1$ and $\B_2$ and returns their conjunction denoted by $\B_1\land \B_2$. 
It proceeds from the root downward creating vertices in the resulting graph
as follows: 
\begin{enumerate}
\item the function $\dec$ 
 decomposes  a non-terminal OBDD node into its constituent components, that is 
its variable and cofactors;

\item  the function $\newnode$ constructs a new OBDD node if it is
not already present, and otherwise returns the already existent node. 
\end{enumerate}
%\begin{center}
\begin{algorithm}\label{algo:and}  \caption{The algorithm for constructing $\B_1\land \B_2$}
 \KwData{$\B_1,\B_2$}
 \KwResult{$\B_1\land \B_2$}
   
   \If {$\B_1=\fa$ or $\B_2=\fa$}
    {return $\fa$\;}

\If {$\B_1=\tr$}
{ return $\B_2$\;}

\If {$\B_2=\tr$} 
{ return $\B_1$\;}

   $(x,\high(\B_1),\low(\B_1))=\dec(\B_1)$;
   
   $(y,\high(\B_2),\low(\B_2))=\dec(\B_2) $;

\If {$x=y$} 
{return  $\newnode(x, \high(\B_1)\wedge \high(\B_2), \low(\B_1)\wedge \low(\B_2))$\;}

\If {$x\tvo y$} 
{return  $\newnode(x,\high(\B_1)\wedge\B_2, \low(\B_1)\wedge\B_2) $\;}

\If {$y\tvo x$} 
{return  $\newnode(y,\B_1\wedge \high(\B_2),\B_1\wedge \low(\B_2)) $\;}

\end{algorithm}
%\end{center}

Lemma \ref{lem:alg1} is a technical lemma which will be used  to prove Corollary \ref{lem:alg}. It follows relatively  straightforwardly from the definition of Algorithm \ref{algo:and}. 

\begin{lemma}\label{lem:alg1} Assume OBDDs $\B_1$ and $\B_2$ such that $\var(\B_1)\neq \emptyset$ and  $\var(\B_2)\neq \emptyset$.   Then Algorithm \ref{algo:and} returns the OBDD $\B_1\land \B_2$ such that 
\begin{itemize}
\item for each $\alpha \in \pathsf(\B_1\land \B_2)$ there is  $\beta\in \pathsf(B_1)\cup \pathsf(B_2)$ such that
$\beta \subseteq \alpha$;
\item for each $\beta\in \pathsf(B_1)\cup \pathsf(B_2)$ there is  $\alpha \in \pathsf(\B_1\land \B_2)$ such that
$\beta \subseteq \alpha$.
\end{itemize}
\end{lemma}

\begin{proof} We give a proof by induction on $k=|\var(\B_1)\cup \var(\B_2)|$.
For the basis step we choose $k=0$ and the lemma trivially holds.

We assume that the lemma holds for any $\B_1'$ and $\B_2'$ such that  $|\var(\B'_1)\cup \var(\B'_2)|\leq k-1$. Let  $|\var(\B_1)\cup \var(\B_2)|=k$ and  $x\in \var(\B_1)\cup \var(\B_2)$ be the smallest variable. 
Then by the  definition of an OBDD 
$|\var(\high(\B_1\wedge \B_2))| \leq k-1$ and
$|\var(\low(\B_1\wedge \B_2))| \leq k-1$.

%$$|\var(\low(\B_1\wedge \B_2))| \leq k-1$$
If  $x\in \var(\B_1)$ and $x\in \var(\B_2)$ then by the definition of Algorithm  \ref{algo:and} it  returns the OBDD $\B_1\wedge \B_2$ such that 
\begin{itemize}
\item $\var(\rt)=x$; 

\item $\high(\B_1\wedge \B_2)=\high(\B_1)\wedge \high(\B_2)$  and $\low(\B_1\wedge \B_2)=\low(\B_1)\wedge \low(\B_2)$.

\end{itemize}
%%%%%%%%%%%%%%%%%%%%%%%

If $x\in \var(\B_1)$ and $x\not\in \var(\B_2)$ then
 by the definition of Algorithm  \ref{algo:and} it  returns  the OBDD $\B_1\wedge \B_2$ such that 
\begin{itemize}
\item $\var(\rt)=x$;
\item $\high(\B_1\wedge \B_2)=\high(\B_1)\wedge \B_2$ and $\low(\B_1\wedge \B_2)=\low(\B_1)\wedge \B_2$.
\end{itemize}
We use the induction hypothesis and conclude that the lemma holds.

\end{proof}        

The following corollary  is a direct consequence of Lemma \ref{lem:alg1} and it will be used later to prove the main result.

\begin{corollary}\label{lem:alg} Assume CNFs $\varphi_1$ and $\varphi_2$ and OBDDs $\B_1$ and $\B_2$ such that  $\varphi_1\less \B_1$ and $\varphi_2\less\B_2$.   Then Algorithm \ref{algo:and} returns the OBDD $\B_1\land \B_2$ such that 
   $\varphi_1\land \varphi_2\less\B_1\land \B_2$.
\end{corollary}

\subsection{Reduction Rules}

There  are two  reduction rules  not affecting the semantics of OBDDs
that can be  used to reduce the size of the OBDDs constructed by  Algorithm \ref{algo:and} defined as follows.

\begin{itemize}
\item  {\bf Merging:} If  $\low(p)\isomBDD\low(q)$ and $\high(p)\isomBDD\high(q)$ for $p,q\in\node(\B)$ then the node $p$ can be removed. Any link to the node $p$  is replaced by a link to the node $q$.

\medskip

 \item {\bf Elimination:} If  $\low(p)\isomBDD\high(p)$  for $p\in\node(\B)$ then the node $p$ can be removed. Any link to the node $p$ is replaced by a link to the root of 
$\high(p)$.  
We write $\B\rightarrow_p \B'$ % (or  $\B\xrightarrow{\RO} \B'$ for simplicity) 
 if $\B'$ is obtained from $\B$ by  eliminating the node $p$.   
 \end{itemize}

%In the following we assume that all isomorphic subOBDDs are merged, that is the merging rule cannot be applied any more. 
We use $\B^{\downarrow}$ to denote the reduced OBDD obtained from $\B$, that is   no reduction rule can be applied to $\B^{\downarrow}$ any more.

\begin{lemma}[Bryant \cite{B1986}]\label{prop:equiv} If $\B_1$ and $\B_2$ are  logically equivalent OBDDs then $\B_1^{\downarrow}\isomBDD \B_2^{\downarrow}$. 
\end{lemma}

The  total time complexity of the algorithm is $O(|\B_1|\times |\B_2|)$. In the worst case,  the upper bound is achieved and  $\B_1\wedge \B_2$ can contain $O(|\B_1|\times |\B_2|)$ nodes.

\begin{theorem} [Bryant \cite{B1986}]\label{theorem:size} Let $\B_1$ and  $\B_2$ be two reduced  OBDDs, that is $\B_1\isomBDD \B_1^{\downarrow}$ and $\B_2\isomBDD \B_2^{\downarrow}$. Then
 the size of $\B_1\wedge\B_2$ is  $O(|\B_1|\times |\B_2|)$, and 
 the number of merging and elimination  steps to compute the reduced OBDD corresponding to $\B_1\wedge\B_2$ is at most  $|\B_1\wedge\B_2|$.
\end{theorem}

\begin{proof}
\begin{enumerate}
%\item[]
\item By definition, any subOBDD of $\B_1\wedge\B_2$ is of the form $\B'_1\wedge\B'_2$, where $\B'_1$ and $\B'_2$ are subOBDDs  of $\B_1$ and $\B_2$ respectively. Hence, the size of $\B_1\wedge\B_2$ is bounded by $O(|\B_1|\times |\B_2|)$.

\item This is immediate since the elimination and merging rules strictly decrease the number of nodes.
\end{enumerate}

\end{proof}

 \subsection{Notations and Technical Background}

Now we  define notations that will be  used in the rest of the paper and introduce some simple technical background  we need  to prove the main result. 
Thus, some additional properties related to the construction of an OBDD by Algorithm \ref{algo:and} are introduced in Lemma \ref{lemma:paths-pairs}  and straightforward combinatorial results   are defined in Lemma \ref{lemma:comb}.

Let   $\varphi$ be a CNF  and $\B$  be an OBDD such that $\varphi\less \B$. 
We  tacitly assume a  function $$\af: \pathsf(\B)\rightarrow \cls(\varphi)$$  such that 
$\alpha\not\models \af(\alpha)$ for an $\alpha\in \pathsf(\B)$.  For each node $p\in\node(\B)$, we define the set
$$\cls(p,\varphi)=\{C\in\cls(\varphi)\mid \exists \alpha\in\pathsf(p): C=\af(\alpha) \}$$
 For each $C \in \cls(\varphi)$ and $p\in\node(\B)$, we define 
$$ \rep(p,C)=|\{\alpha\in \pathsf(p)\mid \af(\alpha)=C \}| -1$$
and
$$\fc(p, \varphi)=\sum_{C\in \cls(p,\varphi)} \rep(p,C)$$

  Suppose 
 $\high(p)\isomBDD \low(p)$  and $\B\rightarrow_{p}\B'$  for a node $p\in \node(\B)$ and an  OBDD $\B'$. 
The set    $\respairs(p,\varphi)$ is defined  as follows:
 $$\respairs(p,\varphi)=\{\psi\in \cnf \mid  \varphi\vdash_{\res} \psi \mbox{ and } \varphi\land \psi\less \B'\}$$

Let $x=\var(p)$ and $q^h\in\node(\high(p))$ and $q^l\in\node(\low(p))$. We write  $$q^h\sim_p q^l$$ 
to denote that  $\alpha.x. \beta \in\pathsf(q^h)$ if and only if
$\alpha.\lnot x. \beta \in\pathsf(q^l)$ for some $\alpha$ and $\beta$.

Let  $\B\rightarrow_p \B'$.  %We write  $q \rightarrow_p q'$ for $q\in\node(\B)$ and $q'\in\node(\B')$ if 
For simplicity,  we define informally  what we mean by  $q\rightarrow_p q'$. The OBDD $\B'$ contains in fact   the same nodes as $\B$ except the node $p$. 
We write $$q\rightarrow_p q'$$ to denote  that the node $q'\in \node(\B')$ is in fact the node  $q\in \node(\B)$ after renaming as $q'$.

In the following we assume that clauses of $\varphi$ are not subsumed while constructing the OBDD encoding it using Algorithm \ref{algo:and}. This is formalised with Lemma \ref{lemma:paths-pairs} below.

\begin{lemma}\label{lemma:paths-pairs}  Assume a CNF $\varphi$ and an OBDD $\B$ such that $\varphi\less \B$.  Let $\high(p)\isomBDD \low(p)$ for some $p\in\node(\B)$ with $x=\var(p)$. %Let $q^h\sim_p q^l$ for $q^h\in\node(\high(p))$ and  $q^l\in\node(\low(p))$ with $y=\var(q^h)=\var(q^l)$. 
Suppose  the following holds:
\begin{enumerate}
\item\label{as1} $C\in\cls(\high(p),\varphi) $ and $D\in\cls(\low(p),\varphi)$;

%\item for any literal $l$, if  $l\in\lit(C')$ then $\lnot l\not \in \lit(D')$  with $C=\lnot x\vee C'$ and $D=x\vee D'$.
\item\label{as3} Let $C=\lnot x\vee C'$ and $D=x\vee D'$. Then for any literal $l$, if 
$l\in\lit(C')$ then $\lnot l\not\in\lit(D')$;
\item \label{as4} There is $x'\in\var(C)\cap\var(D)$ such that for any $x''\in var(C)\cup\var(D)\backslash\{x'\}$, $x''\tvo x'$.

\end{enumerate}
Then for some $\alpha_1,\alpha_2$ there are $\alpha_1.x.\alpha_2 \in \pathsf(\high(p))$ and $\alpha_1.\lnot x.\alpha_2\in \pathsf(\low(p))$ such that 
$$\alpha_1.x.\alpha_2\not \models C$$ $$\alpha_1.\lnot x.\alpha_2 \not\models D$$

\end{lemma}

\begin{proof} 
%Lemma \ref{lem:alg1}   
In fact, the lemma statement is implied directly  by  the assumption $\high(p)\isomBDD\low(p)$  and the definition of   Algorithm \ref{algo:and}. But we will  provide a somewhat formal proof by induction on $k=|\var(\pathsf(p))|$.

Let $k=1$. Now we obtain that  $\alpha_1=\alpha_2=\varepsilon$ where $\varepsilon$ is the empty string and the lemma holds.
 We assume that the lemma holds for $k$ with $k\geq 1$ and show that it hold for $k+1$. We consider the following cases:
 \begin{itemize}

%\item $\alpha_1=\varepsilon$,  and  $\alpha_2\neq \varepsilon$ and $\alpha'_2\neq \varepsilon$.   
 \item The node  $p$ is  not the root of $\B$. Then the lemma holds by the induction hypothesis and the definition of   Algorithm \ref{algo:and}  for  $\high(\rt(\B))$ and $\varphi|_{\lnot \var(\rt(\B))}$ or  for $\low(\rt(\B))$ and $\varphi|_{\var(\rt(\B))}$. Now the lemma holds straightforwardly for $\B$ and $\varphi$.
 
 \item The node $p$ is the root of $\B$.  Let  $q^h=\rt(\high(p))$ and $q^l=\rt(\low(p))$. It follows from $\high(p)\isomBDD \low(p)$ that $\high(q^h)\isomBDD \high(q^l)$ and $\low(q^h)\isomBDD \low(q^l)$. 
  
We construct the OBDD $\B^h$ by redirecting  there true and false branches of $p$ to the $\rt(\high(q^h))$ and  $\rt(\high(q^l))$ correspondingly; and $\B^l$ by redirecting there true and false branches of $p$ to the $\rt(\low(q^h))$ and  $\rt(\low(q^l))$ correspondingly. Let $y=\rt(q^h)=\rt(q^l)$. By the induction hypothesis the lemma holds for $\B^h$ and $\varphi|_{\lnot y}$ or for $\B^l$ and $\varphi|_{y}$. Hence, the lemma holds for arbitrary $\B$ and $\varphi$.
\end{itemize}
 
\end{proof}

 Lemma below  presents some technical results  which will be used to prove Theorem \ref{lemma:invariant}. 

\begin{lemma}\label{lemma:comb}

 Let $S$ be a finite set such that $|S|>0$, and $B_1, \dots, B_l\subseteq S$ be a sequence with:
\begin{itemize}
\item $B_l=S$

\item For each $B_i$, $1\leq i < l$,  one of the following holds:
\begin{itemize}

\item $B_i=\{s\}$ for $s \in S$

\item  $B_i =B_j\cup B_k$ for some $j,k$ with $j< k <i$

\end{itemize}
\end{itemize}
Then $l=2 |S|-1$.

\end{lemma}

\begin{proof} We give a proof by induction on $|S|$. 
As   the basis step we choose $|S|=1$. Then the lemma hold as  trivially $l=1$.
Let  the lemma hold for any $|S|$. Assume a set $S'$  such that $|S'|=|S|+1$.  Then $l'=l+2=(2 |S|-1) +2=2|S'|-1$.

\end{proof}

\section{OBDDs and Resolution as Proof Systems}\label{sec:proof_systems}

Proof systems based on resolution and OBDDs are so-called refutational proof systems. 
A refutation of
an unsatisfiable CNF $\varphi$ starts with the clauses of $\varphi$ and
derives a contradiction represented by the
empty clause $\bot$ for resolution and by the $\fa$ node for OBDDs.

Any   proof system operating with OBDDs can be seen as an instance of the  constraints based proof
system. 
In the following we consider the OBDD proof system which uses two rules, \ax~and \jn.

\begin{definition}[OBDD refutation]\label{def:obdd-ref} An \emph{OBDD refutation}  of a CNF $\varphi$ is a sequence of OBDDs $\B_1, \dots, \B_k$  such that 
the following holds:
\begin{itemize}
\item $\ax :$ $\B_i \isomBDD C_i$ with  $1\leq i\leq |\cls(\varphi)|$;
\item $\jn :$ $\B_{i}\isomBDD\B^{\downarrow}_{j'} \wedge \B^{\downarrow}_{j''}$  with    $1 \leq j'<j''<i$ and $|\cls(\varphi)|< i \leq k$;
\item $\B^{\downarrow}_k\isomBDD \fa$.
\end{itemize}
The  size of the OBDD refutation  is defined as $\sum_{i=1}^k |\B_i|$.
\end{definition}

Without loss of generality we can assume that each OBDD is used exactly once, that is if a CNF $\varphi$ consists of $m$ clauses then the number of OBDDs in the OBDD refutation of $\varphi$ is exactly $2m-1$.

The resolution proof system goes back  to Robinson~\cite{R65} and  consists of a single   rule. 
It derives from two clauses $l\vee C$ and $\lnot l\vee D$,  such that $C$ and $D$ do not contain a complementary literal,  the  new clause $C\vee D$ called the resolvent  of  $l\vee C$ and $\lnot l\vee D$,
and denoted  in the following by  $\res(l\vee C, \lnot l\vee D)$. 

When we write $\res(C,D)$,  we assume that there is a literal  $l$ such that $l\in\lit(C)$ and $\lnot l\in\lit(D)$;  moreover, the clauses $C$ and $D$ contain  no other complimentary literals.

\begin{definition}[Resolution refutation]   A resolution refutation of a CNF $\varphi$ of size $k$  is a sequence of clauses $C_1, \dots, C_k$ such that 

\begin{enumerate}

\item $\ax :$ $C_i\in\cls(\varphi)$ with  $1\leq i\leq |\cls(\varphi)|$;
\item $\jn :$ $C_i=\res(C_{j'},C_{j''})$  with $1 \leq j'<j''<i$ and   $|\cls(\varphi)|< i \leq k$;
\item $C_k=\bot$.
\end{enumerate}

We say that $k$ is  the size of the resolution refutation. 
We write $\varphi \vdash_{\res} \psi$ for any CNF $\psi$ such that $\psi=\bigwedge_{i=|\cls(\varphi)|+1}^{k'} C_i$ with $k'\leq k$.
 \end{definition}

\section{Simulating OBDDs by Resolution}\label{sec:simulation}

In the rest of the paper  we  show formally  that if there is an OBDD refutation  of a CNF $\varphi$ of size $n$ then there is a resolution refutation of $\varphi$ of size at most $n^2$. 
The existence of such  resolution refutation  is based on the following observations:

\begin{enumerate}

\item elimination of a node can be simulated by at most $|\cls(\varphi)|$ resolution steps;
\item  $|\cls(\varphi)|\leq n$;
%\item The number of resolution steps corresponding to removing of one node in the OBDD refutation is bounded by $n$.
\item the number of nodes which can be removed by the elimination rule is at most   $n$.
\end{enumerate}

Our main argument is based on the  idea    that  the elimination rule can be simulated  by applying the resolution rule on the variable corresponding to the eliminated node  \cite{P2008}. 

%But we use it differently. 
We 
strengthen this idea and 
  prove that  the number of resolution steps corresponding to the elimination of a node  is  bounded by the number of clauses  in the input CNF $\varphi$.

Moreover, we show that it is  an  invariant property: although    resolution steps generate new clauses,  the number of resolution steps needed to simulate   elimination of a node in an  intermediate OBDD remains  bounded by the number of clauses  of  $\varphi$ encoded by this OBDD.
%Furthermore, we show that it is sufficient to simulate only the elimination rule, that is  the merging rule plays no role in the context. 
We also show that  it is sufficient to simulate only the elimination rule. 
As the merging rule plays no role in the context, we assume for simplicity that all intermediate OBDDs are merged. 

\begin{example} \normalfont  Before giving  technical details, we provide a simple illustrating example.  We consider the CNF  $$\varphi=(x\vee \lnot y)\wedge (y\vee z)\wedge (y\vee \lnot z) \wedge \lnot x$$ An OBDD refutation of $\varphi$ is depicted in Figure \ref{fig:obdd-ref}. The OBDDs $\B_6^{\downarrow}$ and  $\B_7^{\downarrow}$ are obtained from $\B_6$ and $\B_7$ correspondingly by applying the elimination rule.
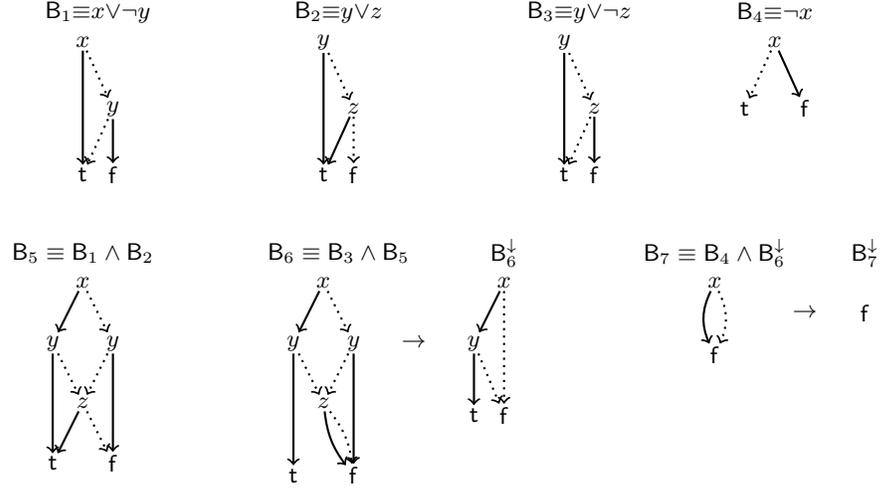
\begin{figure}[t]
\begin{center}
\begin{tikzpicture}
  [scale=0.8,
  inner sep = 1 pt]
  
  \begin{scope}[shift={(0,0)}]
    \node () at (0.25,0.5) {$\B_1{\equiv}x{\lor} \lnot y$};
    \node (x)  at (0,0) {$x$};
    \node (z)  at (0.5,-1.1)  {$y$};
    \node (1)  at (0,-2.2)  {$\T$};
    \node (0)  at (0.5,-2.2)  {$\F$};
   
    \path[thick,->]
    (x) edge             	(1) 
    (z) edge 			(0);
    \path[thick,->,dotted]
    (x) edge             	(z) 
    (z) edge 			(1);
\end{scope}

%%%%%%%%%%%%%%%%%%%

\begin{scope}[shift={(4,0)}]

  \node () at (0.25,0.5) {$\B_2{\equiv}y{\lor}z$};
    \node (y)  at (0,0) {$y$};
    \node (z)  at (0.5,-1.1)  {$z$};
    \node (1)  at (0,-2.2)  {$\T$};
    \node (0)  at (0.5,-2.2)  {$\F$};
    
    \path[thick,->]
    (y) edge             	(1) 
    (z) edge 			(1);
    \path[thick,->,dotted]
    (y) edge             	(z) 
    (z) edge 			(0);
\end{scope}

%%%%%%%%%%%%%%%%%%%

  \begin{scope}[shift={(8,0)}]
    \node () at (0.25,0.5) {$\B_3{\equiv}y{\lor}\neg z$};
    \node (y)  at (0,0) {$y$};
    \node (z)  at (0.5,-1.1)  {$z$};
    \node (1)  at (0,-2.2)  {$\T$};
    \node (0)  at (0.5,-2.2)  {$\F$};
    
    \path[thick,->]
    (y) edge             	(1) 
    (z) edge 			(0);
    \path[thick,->,dotted]
    (y) edge             	(z) 
    (z) edge 			(1);
  \end{scope}

%%%%%%%%%%%%%%%%%%%

  \begin{scope}[shift={(11.5,0)}]
    \node () at (0,0.5) {$\B_{4}{\equiv}\neg x$};
    \node (y)  at (0,0) {$x  $};
    \node (1)  at (-0.5,-1.1)  {$\T$};
    \node (0)  at (0.5,-1.1)  {$\F$};
    
    \path[thick,->]
    (y) edge             	(0); 
    \path[thick,->,dotted]
    (y) edge             	(1);
  \end{scope}

%%%%%%%%%%%%%%%%%%%

\begin{scope}[shift={(0,-4.0)}]
  \node () at (0,0.5) {$\B_5\equiv \B_1\wedge \B_2$};
  \node (x)  at (0,0) {$x$};
  \node (y1) at (-0.5,-1) {$y$};
  \node (y2) at (0.5,-1) {$y$};
  \node (z)  at (0,-2) {$z$};
  \node (1)  at (-0.5,-3)  {$\T$};
  \node (0)  at (0.5,-3)  {$\F$};
  
  \path[thick,->]
  (x) edge             	(y1) 
  (y1) edge             	(1) 
  (y2) edge     	(0) 
  (z) edge 			(1);
  \path[thick,->,dotted]
  (x) edge             	(y2) 
  (y1) edge             	(z) 
  (y2) edge      	(z) 
  (z) edge 			(0);

\end{scope}

%%%%%%%%%%%%%%%%%%%

\begin{scope}[shift={(4,-4.0)}]
  \node () at (0.25,0.5) {$\B_6 \equiv \B_3\wedge \B_5$};
    \node (x)  at (0,0) {$x$};
    \node (y1) at (-0.5,-1) {$y$};
    \node (y2) at (0.5,-1) {$y$};
    \node (z1) at (0,-2) {$z$};
    \node (1) at (-0.5,-3.2)  {$\T$};
    \node (0) at (0.5,-3.2)  {$\F$};
    
    \path[thick,->]
    (x) edge             	(y1) 
    (y1) edge             	(1) 
    (y2) edge            	(0) 
    (z1) edge[bend right=15]	(0);

    \path[thick,->,dotted]
    (x) edge             	(y2) 
    (y1) edge             	(z1) 
    (y2) edge      	        (z1) 
    (z1) edge[bend left=15]    	(0) ;

\node () at (1.5,-1) {$\rightarrow$};

\end{scope}

%%%%%%%%%%%%%%%%%%%

\begin{scope}[shift={(7,-4.0)}]
  \node () at (0,0.5) {$\B^{\downarrow}_6$};
  \node (x)  at (0,0) {$x$};
  \node (y1) at (-0.5,-1) {$y$};
  \node (1)  at (-0.5,-2.2)  {$\T$};
  \node (0)  at (0,-2.2)  {$\F$};
  
  \path[thick,->]
  (x)  edge             	(y1) 
  (y1) edge             	(1) ;

  \path[thick,->,dotted]
  (x)  edge             	(0) 
 (y1) edge             	(0) ;

\end{scope}

%%%%%%%%%%%%%%%%%%%

  \begin{scope}[shift={(10.5,-4.0)}]
    \node () at (0,0.5) {$\B_{7}\equiv \B_4\wedge\B^{\downarrow}_{6}$};
    \node (x)  at (0,0)      {$x$};
    \node (0) at (0,-1.2)  {$\F$};
    
    \path[thick,->]
    (x) edge[bend right=25]   	(0); 
    \path[thick,->,dotted]
    (x) edge[bend left=25]   	(0);

\node () at (1.5,-0.5) {$\rightarrow$};

  \end{scope}

%%%%%%%%%%%%%%%%%%%

  \begin{scope}[shift={(13,-4.0)}]
    \node () at (0,0.5) {$\B^{\downarrow}_{7}$};
    \node (0)  at (0,-0.5)  {$\F$};
  \end{scope}

\end{tikzpicture}

\end{center}
\caption{The OBDD refutation of $\varphi=(x\vee \lnot y)\wedge (y\vee z)\wedge (y\vee \lnot z) \wedge \lnot x$  where the solid lines represent the $\tr$-branches and the dotted lines represent the $\fa$-branches} \label{fig:obdd-ref}
\end{figure}
Consider the following resolution refutation of $\varphi$:
$$C_1=x\vee \lnot y,  C_2=y\vee z, C_3=y\vee \lnot z, C_4= \lnot x,$$
$$ C_5=\res(C_2,C_3)=y, C_6=\res(C_1,C_5)=x, C_7=\res(C_4,C_6)=\bot$$
It  simulates the OBDD refutation as follows:
\begin{itemize}
\item By resolving the clauses $C_2$ and $C_3$ we eliminate the occurrences of $z$. By   resolving the clauses $C_1$ and $C_5$ we eliminate the occurrences of $y$.  The new set of clauses decodes the OBDD $\B_6^{\downarrow}$.

\item  By resolving the clauses $C_4$ and $C_6$ we eliminate the occurrences of $x$ in $\B_7$.
\end{itemize}

\end{example}

\subsection{Simulation of the Elimination Rule}

In this section we show that elimination of a node can be simulated by at most $|\cls(\varphi)|$ resolution steps, where $\varphi$ is the input unsatisfiable CNF.

Lemma  \ref{lemma:res_upperbound} demonstrates 
 that  elimination of a node $p$ can be simulated by  at most $k/2$ resolution steps, where $k$ is the number of paths going through $p$ to the $\fa$-node.  This is a variant  of Lemma 2 from  \cite{P2008} which   serves our needs better.

\begin{lemma}\label{lemma:res_upperbound}
Assume a CNF $\varphi$ and an OBDD $\B$ such that $\varphi\less \B$.  Let
$\B\rightarrow_{p}\B'$  for a node  $p\in \node(\B)$ and an OBDD  $\B'$.  Then  there is a CNF  $\psi \in \respairs(p,\varphi)$
such that 
 $$|\cls(\psi)|~\leq ~|\paths^f(p)|/2 $$
 \end{lemma}

 \begin{proof}
Assume that  $x=\var(p)$.
Let  $\psi$ be the smallest CNF satisfying the following:  
for all  $\alpha$ and $\beta$ such that $\alpha.x.\beta, \alpha.\lnot x.\beta \in \pathsf(p)$, 
$$\lnot (\exists C\in\cls(\varphi):  \alpha.\beta\not\models C) \; \implies \; \res(\af(\alpha.x.\beta),\af(\alpha.\lnot x.\beta))      \in\cls(\psi)$$
 By construction,   
 $\psi \in \respairs(p,\varphi)$ and 
  $|\cls(\psi)|\leq  |\paths^f(p)|/2$.

\end{proof}

\begin{example}  \normalfont
We provide another illustrating  example. Consider the CNF   $\varphi$ consisting of the following  eight clauses:
{
 \begin{align*}
 &C_1=\lnot x\vee \lnot y\vee \lnot v  &C_2=& \lnot x\vee \lnot z \vee  \lnot w  &C_3=&\lnot x\vee y\vee \lnot v  &C_4=& \lnot x\vee z \vee \lnot  w \\
 &D_1= ~  x\vee \lnot z \vee \lnot v     &D_2= & ~  x\vee \lnot y\vee \lnot  w  &D_3= & ~  x\vee z \vee \lnot v  &D_4= & ~ x\vee y \vee \lnot  w 
\end{align*}
}

Figure \ref{fig:obdd1} represents  the OBDD encoding of  $\varphi$,  and the mapping of the $\fa$-paths onto the  clauses of $\varphi$. 
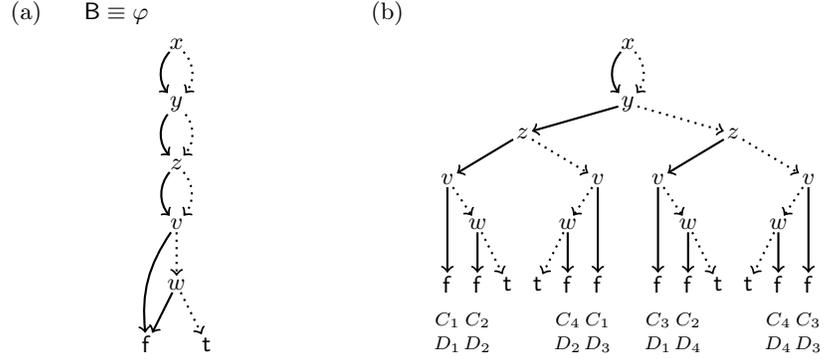
\begin{figure}[t]
 \begin{center}
\begin{tikzpicture}
  [scale=0.8,
  inner sep = 1 pt]

\begin{scope}[shift={(-3.5,0)}]

\node (a)  at (-2.5,0.5) {{(a)}}; 
\node (a)  at (-1,0.5) {{$\B\equiv \varphi$}};

\node (x0)  at (0,0)  {$x$};
\node (x1)  at (0,-1)  {$y$};
\node (x2)  at (0,-2)  {$z$};   
\node (x3) at (0,-3)  {$v$}; 
\node (x4) at (0,-4)  {$w$};

\node (1)  at (0.5,-5)  {$\T$};
\node (0)  at (-0.5,-5)  {$\F$};

\path[thick,->]
(x0)edge[bend right=40]    	(x1)
(x1)edge[bend right=40]    	(x2)
(x2)edge[bend right=40]    	(x3)
(x3)edge[bend right=20]    	(0)

(x4) edge    (0);
      
\path[thick,->,dotted] 
(x0)edge[bend left=40]    	(x1)
(x1)edge[bend left=40]    	(x2)
(x2)edge[bend left=40]    	(x3)

(x3) edge   (x4)   
(x4) edge   (1) ;   
\end{scope}

%%%%%%%%%%%%%%%%%%%%%%

\begin{scope}[shift={(4,0)}]

\node (a)  at (-4,0.5) {{(b)}}; 

\node (x)  at (0,0)  {$x$};

%\draw [<-, dotted, rounded corners] (0.2,-0.9)-- (0.45,-0.5) --(0.2,-0.1);
%\draw [<-, rounded corners] (-0.2,-0.9) -- (-0.45,-0.5) -- (-0.2,-0.1);

\node (x11)  at (0,-1)  {$y$};
  
\node (x21)  at (-1.75,-1.5)  {$z$};
\node (x22)  at (1.75,-1.5)  {$z$};   
     
\node (x31) at (-3,-2.25)  {$v$};
\node (x32) at (-0.5,-2.25)  {$v$};
\node (x33) at (0.5,-2.25)  {$v$};
\node (x34) at (3,-2.25)  {$v$};
 
\node (x41) at (-2.5,-3)  {$w$};
\node (x42) at (-1,-3)  {$w$};
\node (x43) at (1,-3)  {$w$};
\node (x44) at (2.5,-3)  {$w$};

\node (11)  at (-2,-4)  {$\T$};
\node (12)  at (-1.5,-4)  {$\T$};   
\node (13)  at (1.5,-4)  {$\T$};
\node (14)  at (2,-4)  {$\T$};

\node (01)  at (-3,-4)  {$\F$};
\node (02)  at (-2.5,-4)  {$\F$};    
\node (03)  at (-1,-4)  {$\F$};
\node (04)  at (-0.5,-4)  {$\F$};  
  
\node (05)  at (0.5,-4)  {$\F$};   
\node (06)  at (1,-4)  {$\F$};
\node (07)  at (2.5,-4)  {$\F$};    
\node (08)  at (3,-4)  {$\F$};
 
\path[thick,->]
(x) edge[bend right=40]    	(x11)

(x11) edge    (x21) 
(x21) edge    (x31)
(x22) edge    (x33) 
(x31) edge   (01)
(x41) edge    (02)
(x32) edge     (04)      
(x42) edge    (03)
(x33) edge     (05)
(x43) edge    (06)
(x34) edge    (08)
(x44) edge    (07);
         
\path[thick,->,dotted]    
(x) edge[bend left=40]    	(x11)

(x11) edge     (x22)
(x21) edge     (x32)
(x22) edge     (x34)
(x31) edge     (x41)
(x32) edge     (x42)
(x33) edge     (x43)
(x34) edge     (x44)
(x41) edge      (11)  
(x42) edge     (12)  
(x43) edge     (13)  
(x44) edge     (14) ;  

%\node ()  at (-6,-5.8) {\scriptsize{$ C_i=\af(x.\alpha_i)$}};   
%\node ()  at (-6,-6.3) {\scriptsize{$ D_i=\af(\lnot x.\alpha_i)$}};   
           
\node ()  at (-3,-4.6) {\scriptsize{$C_1$}};   
\node ()  at (-3,-5) {\scriptsize{$D_1$}};  
\node ()  at (-2.5,-4.6) {\scriptsize{$C_2$}};   
\node ()  at (-2.5,-5) {\scriptsize{$D_2$}}; 
\node ()  at (-1,-4.6){\scriptsize{$C_4$}}; 
\node ()  at (-1,-5) {\scriptsize{$D_2$}};      

\node ()  at (-0.5,-4.6) {\scriptsize{$C_1$}};  
\node ()  at (-0.5,-5) {\scriptsize{$D_3$}};

\node ()  at (0.5,-4.6) {\scriptsize{$C_3$}}; 
\node ()  at (0.5,-5) {\scriptsize{$D_1$}};

\node ()  at (1,-4.6) {\scriptsize{$C_2$}};   
\node ()  at (1,-5) {\scriptsize{$D_4$}};   

\node ()  at (2.5,-4.6) {\scriptsize{$C_4$}};  
\node ()  at (2.5,-5) {\scriptsize{$D_4$}}; 

\node ()  at (3,-4.6) {\scriptsize{$C_3$}};   
\node ()  at (3,-5) {\scriptsize{$D_3$}};  

\end{scope}

\end{tikzpicture}

 \caption{An example illustrating Theorem \ref{lemma:res_upperbound}:  (a) the OBDD encoding of  $\varphi$;  and (b) the mapping of the $\fa$-paths of the OBDD onto the set of clauses of $\varphi$}
 \label{fig:obdd1}
\end{center}
\end{figure}
Elimination of  the node labelled with $x$ can be simulated by the following resolution steps:
{ 
 \begin{align*} 
&\res(C_1,D_1)= \lnot y \vee \lnot z \vee \lnot v  &\res(C_2,D_2)=& \lnot y \vee z \vee \lnot v   &\res(C_1,D_3)=&  \lnot y \vee \lnot z \vee  \lnot w \\
  &\res(C_4,D_2)=\lnot y\vee z \vee \lnot  w     & \res(C_3,D_1)=& y \vee \lnot z \vee \lnot v        &\res(C_2,D_4)=& \vee z \vee \lnot v \\
    &\res(C_3,D_3)=y \vee\lnot z \vee  \lnot w 
   &\res(C_4,D_4)=&y\vee z \vee \lnot  w
\end{align*}
}
\end{example}

Corollary \ref{lemma:nodes} below follows directly from Theorem \ref{lemma:res_upperbound}.

\begin{corollary} \label{lemma:nodes} Assume an OBDD $\B$ and a CNF $\varphi$ such that $\varphi\less \B$. Suppose $\high(p) \isomBDD \low(p)$ for some $p\in\node(\B)$.   Let  $\ph\subseteq \pathsf(\high(p))$ and $\pl\subseteq \pathsf(\low(p))$ be sets such that $\alpha.x.\beta\in\ph$ if and only if $\alpha.\lnot x.\beta\in \pl$ for some $\alpha$ and $\beta$. 
Suppose  $\phc=\pathsf(\high(p))\backslash \ph$ and  $\plc=\pathsf(\low(p))\backslash \pl$. Suppose for each pair  of paths $(\alpha.x.\beta,\alpha.\lnot x.\beta)\in \phc \times \plc$ there is a pair of paths $(\alpha'.x.\beta',\alpha'.\lnot x.\beta')\in \ph \times \pl$ and clauses $C,D\in\cls(\varphi)$ such that 
\[\alpha.x.\beta,\alpha'.x.\beta'\not\models C\]
\[\alpha.\lnot x.\beta,\alpha'.\lnot x.\beta'\not\models D\]
Let $k=|\phc|$. Then  there is $\psi\in \respairs(p,\varphi)$ such that   $$|\cls(\psi)|~\leq ~  |\pathsf(p)|/2 - k$$
\end{corollary}

\begin{example} \label{ex:php} \label{ex:gphp} \normalfont
The formulas $\php_n$, $n\geq 1$,  encoding the pigeonhole principle   were studied intensively in relation to complexity of different propositional proof systems and  they are  defined as follows.
\begin{align*}
 \php_n = \bigwedge_{i=1}^{n+1}  \bigvee_{j=1}^n p_{ij} \wedge  \bigwedge_{\mathclap{\substack{1\leq i <j\leq n+1\\ 1\leq k\leq n}}} \lnot p_{ik}\vee \lnot p_{jk}
\end{align*}    
    
We build  the formulas $\gphp_n$ by doubling the number of clauses of  $\php_n$: for some new variable $p_0$ $$p_0\vee C,  \lnot p_0\vee C\in \cls(\gphp_n)$$ if and only if $C\in \cls(\php_n)$.    

\begin{figure}[t]
\begin{center}
\begin{tikzpicture}
  [scale=0.85,
  inner sep = 1 pt]

\begin{scope}[shift={(5,0)}]

\node (p0)  at (0,7) {$p_0$}; 
 
\node (p11_1)  at (0,6) {$p_{11}$};   

\node (p12_1)  at (-1.25,5.5) {$p_{12}$};    
\node (p12_2)  at (1.25,5.5) {$p_{12}$};  

\node (p21_1)  at (-2.5,5) {$p_{21}$};  
\node (p21_2)  at (0,5) {$p_{21}$};  
\node (p21_3)  at (2.5,5) {$p_{21}$};   

\node (p22_1)  at (-2.5,4.25) {$p_{22}$};  
\node (p22_2)  at (-1.25,4.5) {$p_{22}$};  
\node (p22_3)  at (1.25,4.5) {$p_{22}$}; 
\node (p22_4)  at (2.5,4.25) {$p_{22}$}; 
  
\node (p31_1)  at (0,3.5) {$p_{31}$};      

\node (p32_1)  at (0,2.75) {$p_{32}$};

\node (f4)  at (0,1.5) {$\Fa$}; 

\path[->,thick]
(p0) edge[bend right=20]    	(p11_1)

(p11_1) edge           (p12_1) 
(p12_1) edge           (p21_1)
(p12_2) edge           (p21_3)

(p21_1) edge            (f4)
(p21_2) edge[bend right=35]              (f4)
(p21_3) edge           (p22_3)

(p22_1) edge  (f4)
(p22_2) edge           (p31_1)
(p22_3) edge  (f4)
(p22_4) edge[bend right=5] (f4)

(p31_1) edge[bend right=33]    	(f4)
(p32_1) edge[bend right=15]    	(f4)
 ;  
 
\path[->,dotted, thick]
(p0) edge[bend left=20]    	(p11_1)
(p11_1) edge           (p12_2) 

(p12_1) edge           (p21_2)
(p12_2) edge[bend right=3]           (f4)

(p21_1) edge           (p22_1) 
(p21_2) edge           (p22_2)
(p21_3) edge           (p22_4)

(p22_1) edge           (p31_1) 
(p22_3) edge           (p31_1)
(p22_4) edge[bend left=5] (f4)

(p31_1) edge           (p32_1)

(p32_1) edge[bend left=15]    	(f4) 
  ;
  
% \draw[thick,dotted] plot [smooth] coordinates {(0.1,0.85) (0.25,0.45) (0.1,0.1)};
%\draw[thick] plot [smooth] coordinates {(-0.1,0.85) (-0.25,0.45) (-0.1,0.1)};
%\draw[thick] plot [smooth] coordinates {(-0.1,1.85) (-0.5,1) (-0.15,0.1)};
%\draw[thick,dotted] plot [smooth] coordinates {(-0.1,2.85) (-0.75,1.45) (-0.2,0.1)};
%\draw[thick] plot [smooth] coordinates {(-0.6,3.85) (-1,2) (-0.25,0.1)};

\end{scope}

\end{tikzpicture}

\end{center}
\caption{The OBDD encoding of  $\gphp_2$ with $\high(p_0)\isomBDD \low(p_0)$} \label{fig:PHP}
\end{figure}
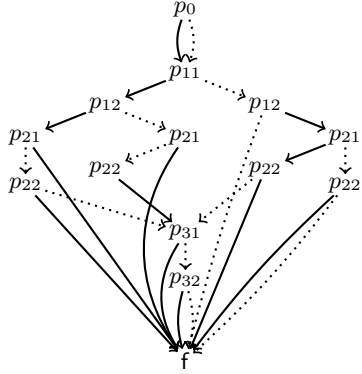

Let $\B$ be the OBDD encoding $\gphp_2$  as it is depicted in Figure \ref{fig:PHP}. %Trivially, $\high(p_0)\isomBDD \low(p_0)$.
Elimination of the node $p_0$ can be  trivially simulated by $|\cls(\php_2)|$ resolution steps. It is sufficient to add  the resolvent $\res(p_0\vee C, \lnot p_0\vee C)=C$.

While the size of the OBDDs encoding $\gphp_n$ will grow exponentially in $n$, elimination of the node labelled with $p_0$ can be simulated in the same manner by the number of resolution steps bounded by $|\cls(\gphp_n)|$ which grows polynomially in $n$.

\end{example}

The subsequent statements improve the upper bound on the number of resolution steps needed to simulate elimination of an arbitrary node.  Namely, we show  that the number of resolution steps sufficient to simulate  elimination of a node $p$  is  bounded by $|\cls(\varphi)|$.

\begin{lemma}\label{lemma:invariant}  Let $\varphi$ be a CNF and $\B$ be an OBDD such that $\varphi\less \B$. Suppose 
$\B\rightarrow_{p}\B'$  for  some $p\in \node(\B)$ and   $\B'$.   Then  there is a CNF $\psi\in \respairs(p,\varphi)$ such that   
 $$|\cls(\psi)|~\leq ~ |\cls(\varphi)|$$
\end{lemma} 

\begin{proof} 
If   $ |\cls(\varphi)|\geq |\pathsf(p)|/2$ then the theorem  holds by Theorem  \ref{lemma:res_upperbound}.
We assume that   $$ |\cls(\varphi)| <|\pathsf(p)|/2$$

 Let    $m=|\pathsf(p)|/2$  and  $m'=|\cls(\varphi)|$ with $m'-m>0$; and   $p^h=\rt(\high(p))$, $p^l=\rt(\low(p))$  and  $x=\var(p)$. 
 %In the following we treat nodes reachable by different paths as distinct nodes, that is that there are  $k=|\pathsf(p)|-1$ distinct  pairs of nodes $Q=\{(q^h_1,q^l_1), \dots, (q^h_k,q^l_k)\}$ such that $q^h_i\sim_p q^l_i$ (using Lemma \ref{lemma:comb}). 
Let $P$ be the following set $$P=\{(\alpha.x.\beta, \alpha.\lnot x.\beta)\mid \exists \alpha,\beta: \alpha.x.\beta \in\pathsf(\high(p)), \alpha.\lnot x.\beta \in\pathsf(\low(p))\}$$
 Assume a function $\af': P\rightarrow \cls(\varphi)$ such that $\af'((\alpha.x.\beta, \alpha.\lnot x.\beta))=C$ if $\alpha.x.\beta\not \models C$ or 
 $\alpha.\lnot x.\beta\not \models C$ and $\bigcup_{(\alpha.x.\beta, \alpha.\lnot x.\beta)\in P}\af'((\alpha.x.\beta, \alpha.\lnot x.\beta))=\cls(p, \varphi)$.
 
 \begin{enumerate}
 \item We define the sets $S, S_1, \dots, S_l$  with $l=2m-1$ as follows: 
 \smallskip
\begin{itemize}
\item $S=P$;

\item $S_i=\{ s_i\}$ for $s_i\in S'$ with $1\leq i\leq m$;

\item $S_i=S'_j\cup S_k$ with $i>j>k$ and $m+1\leq i\leq 2m-1$.
\end{itemize}
\smallskip

 \item We define the sets $S', S'_1, \dots, S'_l$  with $l=2m-1$ as follows: 
 \smallskip
\begin{itemize}
\item $S'=\cls(p^h,\varphi) $;

\item $S_i'=\{ s_i'\}$ for $s_i'\in S'$ with $1\leq i\leq m$;

\item $S'_i=S'_j\cup S'_k$ with $i>j>k$ and $m+1\leq i\leq 2m-1$.
\end{itemize}
\smallskip
Moreover, we assume that  $S'_i=\bigcup{_{(\alpha.x.\beta, \alpha.\lnot x.\beta)\in S_i}\af'((\alpha.x.\beta, \alpha.\lnot x.\beta))}\cap \cls(p^h, \cls\varphi)$.

\smallskip
\item We define the sets $S'', S''_1, \dots, S''_l$  with $l=2m-1$ as follows: 
 \smallskip
\begin{itemize}
\item $S''=\cls(p^l,\varphi) $;

\item $S_i''=\{ s_i''\}$ for $s_i''\in S''$ with $1\leq i\leq m$;

\item $S''_i=S''_j\cup S''_k$ with $i>j>k$ and $m+1\leq i\leq 2m-1$.
\end{itemize}
\smallskip
Moreover, we assume that  $S'_i=\bigcup{_{(\alpha.x.\beta, \alpha.\lnot x.\beta)\in S_i}\af'((\alpha.x.\beta, \alpha.\lnot x.\beta))}\cap \cls(p^l, \cls\varphi)$.
\end{enumerate}
 
 Let  $A_{i}=S'_{i}\cup S''_{i}$ with $1\leq i\leq 2m-1$ and ${\cal{A}}=\{A_1, \dots, A_{2m-1}\}$.
  Now it follows from  Lemma \ref{lemma:comb} and the definition of an OBDD  that there is a set  ${\overline{\cal{A}}}=\{A_{i_1}, \dots, A_{i_{m-m'}}\}$ such that for each $\overline{A} \in {\overline{\cal{A}}}$ there is $A\in {\cal{A}}$ such that $\overline{A} \subseteq A$.  
Hence, by  Lemma \ref {lemma:paths-pairs}  there are $m-m'$ pairs 
 $$(\alpha_1.x.\beta_1, \alpha_1.\lnot x.\beta_1), \dots, (\alpha_m.x.\beta_m, \alpha_m.\lnot x.\beta_m)$$ in the set $P$, let us call this set $P'$,  such that  for each $(\alpha'.x.\beta', \alpha'.\lnot x.\beta')\in P'$ there is $(\alpha.x.\beta, \alpha.\lnot x.\beta)\in P\backslash P'$ such that 
 $\alpha.x.\beta, \alpha'.x.\beta' \not\models C$  and   $\alpha.\lnot x.\beta, \alpha'.\lnot x.\beta' \not\models D$ for some $C,D\in\cls(\varphi)$.     
By  Corollary \ref{lemma:nodes} we obtain that there is a CNF $\psi\in \respairs(p,\varphi)$ such that   
 $|\cls(\psi)|~\leq ~ |\cls(\varphi)|$.

\end{proof}

It follows  from Lemmas \ref{lemma:res_upperbound} and \ref{lemma:invariant}  that  the number of resolution steps needed to simulate elimination of node $p$ is bounded by  the minimum of  $|\pathsf(p)|/2$ and 
$|\varphi)|$.

Lemma \ref{cor:invariant} below is a straightforward consequence of Lemma \ref{lemma:invariant},  and it somewhat relates the upper bound on the number of resolution steps simulating elimination of $p$ and  the number of the $\fa$-paths that go through the node $p$ in combination with  the number of the clauses falsifying these paths (expressed by $\fc(.,.)$). 

\begin{lemma} \label{cor:invariant}  Let $\varphi$ be a CNF and $\B$ be an OBDD such that $\varphi\less \B$. Suppose 
$\B\rightarrow_{p}\B'$  for some  $p\in \node(\B)$ and   $\B'$.   Then  there is a $\psi\in \respairs(p,\varphi)$ such that   
 $$|\cls(\psi)|~  \leq ~ |\pathsf(p)|-\fc(p,\varphi)$$
% with $p^h=\high(p)$ and $p^l=\low(p)$.
\end{lemma}

\begin{proof} We take into account that by definition of $\fc(.,.)$, 
$|\pathsf(p)|-\fc(p,\varphi)=|\cls(p,\varphi)|$,
and the lemma trivially holds.

\end{proof}

\subsection{Invariant}

Now we will prove that although    resolution steps generate new clauses,   the number of resolution steps needed to simulate  elimination of a node  remains  bounded by  the number of clauses   encoded by this OBDD. In fact, we demonstrate a kind of monotonicity 
 expressed by Lemma  \ref{lemma:mon1}.

%\begin{lemma}\label{lemma:mon1} Assume a CNF  $\varphi$  and an OBDD $\B$  with $\varphi\less \B$.  Suppose 
%$\B\rightarrow_{p}\B'$   and  $q\rightarrow_{p}q'$ for $p,q \in\node(\B)$ and $q'\in\node(\B')$.  %Let $\high(q')\isomBDD \low(q')$.
%Then there is a  $\varphi' \in \respairs(p,\varphi)$ such that 
%$$|\pathsf(q')|- \fc(q',\varphi\wedge \varphi')  \leq |\pathsf(q)|-\fc(q, \varphi)$$
%\end{lemma}

\begin{lemma}\label{lemma:mon1} Assume a CNF  $\varphi$  and an OBDD $\B$  with $\varphi\less \B$.  Suppose 
$\B\rightarrow_{p}\B'$   and  $q\rightarrow_{p}q'$ for $p,q \in\node(\B)$ and $q'\in\node(\B')$.  Let $\high(q')\isomBDD \low(q')$.
Then there is a  $\varphi' \in \respairs(p,\varphi)$ such that 
$$|\pathsf(q')|- \fc(q',\varphi\wedge \varphi')  \leq |\pathsf(q)|-\fc(q, \varphi)$$
\end{lemma}

%\begin{proof}  Suppose the nodes  $p$ and  $q$ are not connected by a path. Then removing the node $p$ does not affect the $\fa$-paths that go through the node $q$. That is,   $\pathsf(q')=\pathsf(q)$ and $  \fc(q',\varphi\wedge \varphi')=\fc(q, \varphi)$. Hence, 
%$$|\pathsf(q')|- \fc(q',\varphi\wedge \varphi')  =|\pathsf(q')|- \fc(q',\varphi)  = |\pathsf(q)|-\fc(q, \varphi)$$
%
%
%Now we assume that $p$ and $q$ are connected by a path.  We construct $\varphi'$ as it is defined in the proof of  Lemma \ref{lemma:res_upperbound} and  we consider following cases.
%
%\begin{itemize}
%
%\item Let $q\in\node(\B^p)$ where $\B^p$ is the subOBDD of $\B$ rooted at the node $p$. % $\var(p)\tvo \var(q)$.  
%We observe that $$\alpha.\beta\not\models \res(C,D)$$ if and only if  $\alpha.x.\beta  \not\models C$ and $\alpha.\lnot x.\beta\not\models D$ for some  $\alpha.\beta \in \pathsf(\high(p))$ (alternatively, we could use $ \pathsf(\low(p))$ as   $\high(p)
%\isomBDD \low(p)$) and $C\in\cls(\rt(\high(p)), \varphi)$, $D\in\cls(\rt(\low(p)), \varphi)$. 
%Hence, 
%$$|\pathsf(q')|- \fc(q',\varphi\wedge \varphi')  \leq |\pathsf(q)|-\fc(q, \varphi)$$
%
%
%
%\item Let $p\in\node(\B^q)$ where $\B^q$ is the subOBDD of $\B$ rooted at the node $q$. 
%
%\end{itemize}
%
%\end{proof}

\begin{proof}  %We consider the following cases.
\begin{enumerate}
\item  \label{inv1}  Suppose the nodes  $p$ and  $q$ are not connected by a path. Then removing the node $p$ does not affect the $\fa$-paths that go through the node $q$. That is,   $\pathsf(q')=\pathsf(q)$ and $  \fc(q',\varphi\wedge \varphi')=\fc(q, \varphi)$. Hence, 
$$|\pathsf(q')|- \fc(q',\varphi\wedge \varphi')  =|\pathsf(q')|- \fc(q',\varphi)  = |\pathsf(q)|-\fc(q, \varphi)$$

Now we assume that $p$ and $q$ are connected by a path.  We construct $\varphi'$ as it is defined in the proof of  Lemma \ref{lemma:res_upperbound} and  we consider following cases.

\item \label{inv2} Let $q\in\node(\B^p)$ where $\B^p$ is the subOBDD of $\B$ rooted at the node $p$. % $\var(p)\tvo \var(q)$.  

We observe that $$\alpha.\beta\not\models \res(C,D)$$ if and only if  $\alpha.x.\beta  \not\models C$ and $\alpha.\lnot x.\beta\not\models D$ for some  $\alpha.\beta \in \pathsf(\high(p))$ (alternatively, we could use $ \pathsf(\low(p))$ as   $\high(p)
\isomBDD \low(p)$) and $C\in\cls(\rt(\high(p)), \varphi)$, $D\in\cls(\rt(\low(p)), \varphi)$. 
Hence, 
$$|\pathsf(q')|- \fc(q',\varphi\wedge \varphi')  \leq |\pathsf(q)|-\fc(q, \varphi)$$

\item \label{inv3}  Let $p\in\node(\B^q)$ where $\B^q$ is the subOBDD of $\B$ rooted at the node $q$.  %We consider two cases: 

Let $\overline{\pathsf}(p)=\pathsf(q)\backslash \pathsf(p)$. Let 
$${\clsset}^{p}=\{C\in\cls(p,\varphi)\mid \lnot\exists \alpha\in \pathsf(q)\backslash \pathsf(p): \alpha\not\models C\}$$
$${\clsset}^{q}=\{C\in\cls(q,\varphi)\mid \lnot\exists \alpha\in \pathsf(p): \alpha\not\models C\}$$
$${\clsset}^{q,p}=\{C\in\cls(q,\varphi)\mid \exists \alpha\in \pathsf(p), \beta\in \pathsf(q)\backslash\pathsf(p): \alpha, \beta\not\models C\}$$
We use the same arguments as in case \ref{inv1} for the clauses in ${\clsset}^{q}$, the same arguments as in case \ref{inv2} for the clauses in ${\clsset}^{p}$ and apply 
  Lemma \ref{lemma:paths-pairs} for the clauses in ${\clsset}^{q,p}$.
\end{enumerate}

\end{proof}

\begin{lemma}\label{lemma:con} Assume CNFs $\varphi_1$ and $\varphi_2$, and OBDDs $\B_1$ and $\B_2$ with $\varphi_1\less \B_1$ and  $\varphi_2\less \B_2$.
  Let for any $q_1\in \node(B_1)$ and $q_2\in \node(B_2)$, and some  $k_1, k_2\geq 0$
\begin{itemize}
\item $|\pathsf(q_1)|- \fc(q_1,\varphi_1)\leq k_1 $

\item $|\pathsf(q_2)|- \fc(q_2,\varphi_2)\leq k_2 $
\end{itemize}
Then Algorithm \ref{algo:and} returns the OBDD $\B_1\land \B_2$ such that for any $q\in \node(\B_1\wedge \B_2)$
$$|\pathsf(q)|- \fc(q,\varphi_1\wedge \varphi_2)\leq k_1+k_2$$

\end{lemma}

\begin{proof}  
Observe that for any CNF $\varphi$ and an OBDD $\B$ such that $\varphi\less \B$ and $q\in \node(\B)$, 
$$|\pathsf(q)|- \fc(q,\varphi) \leq |\pathsf(\B)|- \fc(\rt(\B),\varphi)$$

We recall that by Lemma \ref{lem:alg}, 
   $\varphi_1\land \varphi_2\less\B_1\land 
\B_2$ and define the sets $S_1$ and $S_2$ as follows:
\begin{itemize}
\item $S_1=\{\alpha \in \pathsf(\B_1\wedge \B_2) \mid \af(\alpha)\in\cls(\varphi_1)\}$;

\item $S_2=\{\alpha \in \pathsf(\B_1\wedge \B_2) \mid \af(\alpha)\in\cls(\varphi_2)\}$.
\end{itemize}
%
% $$S_1=\{\alpha \in \pathsf(\B_1\wedge \B_2) \mid \af(\alpha)\in\cls(\varphi_1)\}$$
% $$S_2=\{\alpha \in \pathsf(\B_1\wedge \B_2) \mid \af(\alpha)\in\cls(\varphi_2)\}$$
That is,  the set $S_1$ contains the $\fa$-paths of $\B_1\wedge \B_2$ falsified by the clauses of $\varphi_1$ and  $S_2$ contains the $\fa$-paths of $\B_1\wedge \B_2$ falsified by the clauses of $\varphi_2$.
Suppose
\begin{itemize}
\item $m_1=|S_1|-|\pathsf(\rt(\B_1))| $;
\item $m_2=|S_2|-|\pathsf(\rt(\B_2))|$.
\end{itemize}
It follows from Lemma  \ref{lem:alg1} and the definition of $\fc(.,.)$ that 
$$ \fc(\rt(\B_1\wedge \B_2),\varphi_1\wedge \varphi_2) =  \fc(\rt(\B_1),\varphi_1) +  \fc(\rt(\B_2),\varphi_2)+m_1+m_2$$

and therefore for any $q\in\node(\varphi_1\wedge \varphi_2)$
\begin{equation*}
 \begin{aligned}
      |\pathsf(q)|- \fc(q,\varphi_1\wedge \varphi_2)& \leq |\pathsf(\B_1\wedge \B_2)|- \fc(\rt(\B_1\wedge\B_2),\varphi_1\wedge \varphi_2)\\                                                                                  &= (|\pathsf(\B_1)| + |\pathsf(\B_2)|+ m_1+m_2) - & \\
                                                                                 & ~~~~  (\fc(\rt(\B_1),\varphi_1) +  \fc(\rt(\B_2),\varphi_2)+ m_1 +m_2) & \\
                                                                                 &\leq  k_1+k_2 &
 \end{aligned}
\end{equation*}

\end{proof}

Now we combine the results established by  Lemmas \ref{lemma:invariant}-\ref{lemma:con} and obtain the following corollary.

\begin{corollary}\label{cor:ub}  Assume an unsatisfiable CNF $\varphi$. Let $\B_1, \dots, \B_k$  be  an OBDD refutation of $\varphi$. Then elimination of a node in any OBDD $\B_i$, $1\leq i\leq k$,  can be simulated by at most $|\cls(\varphi)|$ resolution steps.
\end{corollary}
%\begin{proof} For any OBDD $\B_i$, $1\leq i \leq k$, there is a subset of clauses of $\varphi$ this OBDD encodes. 
%
%
%\end{proof}

\begin{example}  \normalfont
The CNFs $\php_n$, $n\geq 1$,  formalising the pigeonhole principle is presented in Example \ref{ex:gphp}. 
We consider the OBDD refutation of  $\php_2$  depicted in Figure \ref{fig:php2}. 
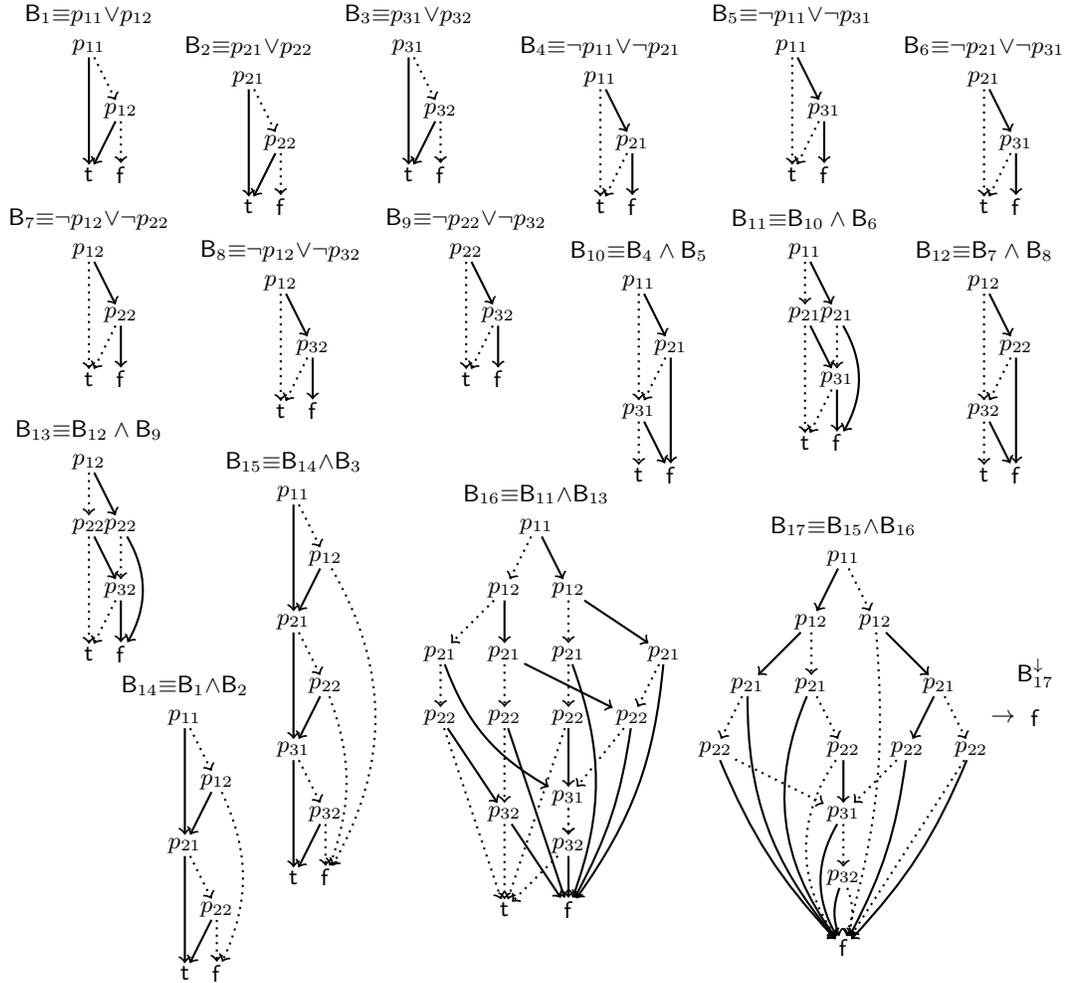
\begin{figure}[t]
 \begin{center}
\begin{tikzpicture}
  [scale=0.85,
  inner sep = 1 pt]             

\begin{scope}[shift={(-8,0.5)}]

%\node (a)  at (-1.5,2.5) {\bf a)}; 

  \node () at (0,2.5) {$\B_1{\equiv}p_{11}{\lor} p_{12}$};

\node (p11) at (0,2) {$p_{11}$};
\node (p12) at (0.5,1) {$p_{12}$};
\node (t)  at (0,0) {$\Tr$}; 
\node (f)  at (0.5,0) {$\Fa$}; 

\path[->, thick]
(p11) edge  (t) 
(p12) edge (t)
;  
 
\path[->,dotted, thick]
(p11) edge  (p12) 
(p12) edge  (f)
;
\end{scope}

\begin{scope}[shift={(-5.5,0)}]

\node () at (0,2.5) {$\B_2{\equiv}p_{21}{\lor} p_{22}$};

\node (p11) at (0,2) {$p_{21}$};
\node (p12) at (0.5,1) {$p_{22}$};
\node (t)  at (0,0) {$\Tr$}; 
\node (f)  at (0.5,0) {$\Fa$}; 

\path[->, thick]
(p11) edge  (t) 
(p12) edge (t)
;  
 
\path[->,dotted, thick]
(p11) edge  (p12) 
(p12) edge  (f)
;
\end{scope}

\begin{scope}[shift={(-3,0.5)}]

\node () at (0,2.5) {$\B_3{\equiv}p_{31}{\lor} p_{32}$};

\node (p11) at (0,2) {$p_{31}$};
\node (p12) at (0.5,1) {$p_{32}$};
\node (t)  at (0,0) {$\Tr$}; 
\node (f)  at (0.5,0) {$\Fa$}; 

\path[->, thick]
(p11) edge  (t) 
(p12) edge (t)
;  
 
\path[->,dotted, thick]
(p11) edge  (p12) 
(p12) edge  (f)
;
\end{scope}

\begin{scope}[shift={(0,0)}]

\node () at (0,2.5) {$\B_4{\equiv}\lnot p_{11}{\lor} \lnot p_{21}$};

\node (p11) at (0,2) {$p_{11}$};
\node (p12) at (0.5,1) {$p_{21}$};
\node (t)  at (0,0) {$\Tr$}; 
\node (f)  at (0.5,0) {$\Fa$}; 

\path[->,dotted, thick]
(p11) edge  (t) 
(p12) edge (t)
;  
 
\path[->, thick]
(p11) edge  (p12) 
(p12) edge  (f)
;
\end{scope}

\begin{scope}[shift={(3,0.5)}]

\node () at (0,2.5) {$\B_5{\equiv}\lnot p_{11}{\lor} \lnot p_{31}$};

\node (p11) at (0,2) {$p_{11}$};
\node (p12) at (0.5,1) {$p_{31}$};
\node (t)  at (0,0) {$\Tr$}; 
\node (f)  at (0.5,0) {$\Fa$}; 

\path[->,dotted, thick]
(p11) edge  (t) 
(p12) edge (t)
;  
 
\path[->, thick]
(p11) edge  (p12) 
(p12) edge  (f)
;
\end{scope}

\begin{scope}[shift={(6,0)}]

\node () at (0,2.5) {$\B_6{\equiv}\lnot p_{21}{\lor} \lnot p_{31}$};

\node (p11) at (0,2) {$p_{21}$};
\node (p12) at (0.5,1) {$p_{31}$};
\node (t)  at (0,0) {$\Tr$}; 
\node (f)  at (0.5,0) {$\Fa$}; 

\path[->,dotted, thick]
(p11) edge  (t) 
(p12) edge (t)
;  
 
\path[->, thick]
(p11) edge  (p12) 
(p12) edge  (f)
;
\end{scope}

\begin{scope}[shift={(-8,-2.7)}]

\node () at (0,2.5) {$\B_7{\equiv}\lnot p_{12}{\lor} \lnot p_{22}$};

\node (p11) at (0,2) {$p_{12}$};
\node (p12) at (0.5,1) {$p_{22}$};
\node (t)  at (0,0) {$\Tr$}; 
\node (f)  at (0.5,0) {$\Fa$}; 

\path[->,dotted, thick]
(p11) edge  (t) 
(p12) edge (t)
;  
 
\path[->, thick]
(p11) edge  (p12) 
(p12) edge  (f)
;
\end{scope}

\begin{scope}[shift={(-5,-3.2)}]

\node () at (0,2.5) {$\B_8{\equiv}\lnot p_{12}{\lor} \lnot p_{32}$};

\node (p11) at (0,2) {$p_{12}$};
\node (p12) at (0.5,1) {$p_{32}$};
\node (t)  at (0,0) {$\Tr$}; 
\node (f)  at (0.5,0) {$\Fa$}; 

\path[->,dotted, thick]
(p11) edge  (t) 
(p12) edge (t)
;  
 
\path[->, thick]
(p11) edge  (p12) 
(p12) edge  (f)
;
\end{scope}

\begin{scope}[shift={(-2.1,-2.7)}]

\node () at (0,2.5) {$\B_9{\equiv}\lnot p_{22}{\lor} \lnot p_{32}$};

\node (p11) at (0,2) {$p_{22}$};
\node (p12) at (0.5,1) {$p_{32}$};
\node (t)  at (0,0) {$\Tr$}; 
\node (f)  at (0.5,0) {$\Fa$}; 

\path[->,dotted, thick]
(p11) edge  (t) 
(p12) edge (t)
;  
 
\path[->, thick]
(p11) edge  (p12) 
(p12) edge  (f)
;
\end{scope}

\begin{scope}[shift={(-6.5,-10)}]

\node () at (0,2.5) {$\B_{14}{\equiv}\B_1{\land} \B_2$};

\node (p11) at (0,2) {$p_{11}$};
\node (p12) at (0.5,1) {$p_{12}$};
\node (p21) at (0,0) {$p_{21}$};
\node (p22) at (0.5,-1) {$p_{22}$};
\node (t)  at (0,-2) {$\Tr$}; 
\node (f)  at (0.5,-2) {$\Fa$}; 

\path[->, thick]
(p11) edge  (p21) 
(p12) edge (p21)
(p21) edge  (t) 
(p22) edge (t)
;  
 
\path[->,dotted, thick]
(p11) edge  (p12) 
(p12) edge[bend left=25]             (f)
(p21) edge  (p22) 
(p22) edge  (f)                                                                                                                       
;

%\draw[dotted] plot [smooth] coordinates {(0.6,0.85) (0.9,-0.5) (0.6,-1.85)};

\end{scope}

\begin{scope}[shift={(-4.8,-6.5)}]

\node () at (0,2.5) {$\B_{15}{\equiv}\B_{14}{\land} \B_3$};

\node (p11) at (0,2) {$p_{11}$};
\node (p12) at (0.5,1) {$p_{12}$};
\node (p21) at (0,0) {$p_{21}$};  
\node (p22) at (0.5,-1) {$p_{22}$};
\node (p31) at (0,-2) {$p_{31}$};
\node (p32) at (0.5,-3) {$p_{32}$};
\node (t)  at (0,-4) {$\Tr$}; 
\node (f)  at (0.5,-4) {$\Fa$}; 

\path[->, thick]
(p11) edge  (p21) 
(p12) edge (p21)
(p21) edge  (p31) 
(p22) edge (p31)
(p31) edge  (t)
(p32) edge (t)
;  
 
\path[->,dotted, thick]
(p11) edge  (p12) 
(p12) edge[bend left=30]             (f)
(p21) edge  (p22) 
(p22) edge[bend left=20]             (f)
(p31) edge  (p32) 
(p32) edge  (f)
;

\end{scope}

\begin{scope}[shift={(0.6,-3.2)}]

\node () at (0,2.5) {$\B_{10}{\equiv}\B_4\wedge \B_5$};
\node (p11) at (0,2) {$p_{11}$};
\node (p21) at (0.5,1) {$p_{21}$};
\node (p31) at (0,0) {$p_{31}$};
\node (t)  at (0,-1) {$\Tr$}; 
\node (f)  at (0.5,-1) {$\Fa$}; 

\path[->, thick]
(p11) edge  (p21) 
(p21) edge  (f) 
(p31) edge (f)
;  
 
\path[->,dotted, thick]
(p11) edge  (p31) 
(p21) edge  (p31) 
(p31) edge  (t)
;

\end{scope}

\begin{scope}[shift={(6,-3.2)}]
\node () at (0,2.5) {$\B_{12}{\equiv}\B_7\wedge \B_8$};
\node (p11) at (0,2) {$p_{12}$};
\node (p21) at (0.5,1) {$p_{22}$};
\node (p31) at (0,0) {$p_{32}$};
\node (t)  at (0,-1) {$\Tr$}; 
\node (f)  at (0.5,-1) {$\Fa$}; 

\path[->, thick]
(p11) edge  (p21) 
(p21) edge  (f) 
(p31) edge (f)
;  
 
\path[->,dotted, thick]
(p11) edge  (p31) 
(p21) edge  (p31) 
(p31) edge  (t)
;
\end{scope}

\begin{scope}[shift={(3.2,-2.7)}]

\node () at (0,2.5) {$\B_{11}{\equiv}\B_{10}\wedge \B_6$};

\node (p11) at (0,2) {$p_{11}$};
\node (p21) at (0,1) {$p_{21}$};
\node (p21_2) at (0.5,1) {$p_{21}$};
\node (p31) at (0.5,0) {$p_{31}$};
\node (t)  at (0,-1) {$\Tr$}; 
\node (f)  at (0.5,-1) {$\Fa$}; 

\path[->, thick]
(p11) edge  (p21_2) 
(p21) edge  (p31) 
(p21_2) edge[bend left=30]             (f)
(p31) edge (f)
;  
 
\path[->,dotted, thick]
(p11) edge  (p21) 
(p21_2) edge  (p31) 
(p21) edge  (t) 
(p31) edge  (t)
;
\end{scope}

\begin{scope}[shift={(-8,-6)}]

\node () at (0,2.5) {$\B_{13}{\equiv}\B_{12}\wedge \B_9$};
\node (p11) at (0,2) {$p_{12}$};
\node (p21) at (0,1) {$p_{22}$};
\node (p21_2) at (0.5,1) {$p_{22}$};
\node (p31) at (0.5,0) {$p_{32}$};
\node (t)  at (0,-1) {$\Tr$}; 
\node (f)  at (0.5,-1) {$\Fa$}; 

\path[->, thick]
(p11) edge  (p21_2) 
(p21) edge  (p31) 
(p21_2) edge[bend left=30]             (f)
(p31) edge (f)
;  
 
\path[->,dotted, thick]
(p11) edge  (p21) 
(p21_2) edge  (p31) 
(p21) edge  (t) 
(p31) edge  (t)
;

%\draw plot [smooth] coordinates {(0.6,0.85) (0.9,0) (0.6,-0.85)};
\end{scope}

\begin{scope}[shift={(-1,-7)}]

\node () at (0,2.5) {$\B_{16}{\equiv}\B_{11}{\land} \B_{13}$};

\node (p11) at (0,2) {$p_{11}$};
\node (p12_1) at (-0.5,1) {$p_{12}$};
\node (p12_2) at (0.5,1) {$p_{12}$};
\node (p21_1) at (-1.5,0) {$p_{21}$};
\node (p21_2) at (-0.5,0) {$p_{21}$};
\node (p21_3) at (0.5,0) {$p_{21}$};
\node (p21_4) at (2,0) {$p_{21}$};
\node (p22_1) at (0.5,-1) {$p_{22}$};
\node (p22_2) at (1.5,-1) {$p_{22}$};
\node (p22_3) at (-0.5,-1) {$p_{22}$};
\node (p22_4) at (-1.5,-1) {$p_{22}$};
\node (p31) at (0.5,-2.25) {$p_{31}$};
\node (p31_1) at (-0.5,-2.5) {$p_{32}$};
\node (p32) at (0.5,-3) {$p_{32}$};
\node (t)  at (-0.5,-4) {$\Tr$}; 
\node (f)  at (0.5,-4) {$\Fa$}; 

\path[->, thick]
(p11) edge  (p12_2) 
(p12_2) edge  (p21_4) 
(p12_1) edge  (p21_2) 
(p21_1) edge[bend right=20]             (p31)
(p21_2) edge  (p22_2)
(p21_3) edge[bend left=20]             (f)
(p21_4) edge[bend left=15]             (f)
(p22_1) edge  (p31)
(p22_2) edge[bend left=10]       (f)
(p22_3) edge   (f)
(p22_4) edge  (p31_1)
(p31_1) edge  (f)
(p32) edge  (f)
;  
 
\path[->,dotted, thick]
(p11) edge  (p12_1) 
(p12_1) edge  (p21_1) 
(p12_2) edge  (p21_3)
(p21_1) edge  (p22_4)
(p21_2) edge  (p22_3)
(p21_3) edge  (p22_1)
(p21_4) edge  (p22_2)
(p22_1) edge  (t)
(p22_2) edge  (p31)
(p31) edge  (p32)
(p22_3) edge  (p31_1)
(p22_4) edge  (t)
(p31) edge  (p32)
(p31_1) edge  (t)
(p32) edge  (t)
;
\end{scope}

\begin{scope}[shift={(3.8,-11.5)}]

%\node (a)  at (-2,6) {\bf a)}; 

\node () at (0,6.5) {$\B_{17}{\equiv}\B_{15}{\land} \B_{16}$};
 
\node (p11_1)  at (0,6) {$p_{11}$};   
\node (p12_1)  at (-0.5,5) {$p_{12}$};    
\node (p12_2)  at (0.5,5) {$p_{12}$};  
\node (p21_1)  at (-1.5,4) {$p_{21}$};  
\node (p21_2)  at (-0.5,4) {$p_{21}$};  
\node (p21_3)  at (1.5,4) {$p_{21}$};   
\node (p22_1)  at (-2,3) {$p_{22}$};  
\node (p22_2)  at (0,3) {$p_{22}$};  
\node (p22_3)  at (1,3) {$p_{22}$}; 
\node (p22_4)  at (2,3) {$p_{22}$};   
\node (p31_1)  at (0,2) {$p_{31}$};      
\node (p32_1)  at (0,1) {$p_{32}$};

\node (f)  at (0,-0.1) {$\Fa$}; 

\path[->,thick]
(p11_1) edge           (p12_1) 
(p12_1) edge           (p21_1)
(p12_2) edge           (p21_3)
(p21_1) edge[bend right=20]             (f)
(p21_2) edge[bend right=30]             (f)
(p21_3) edge           (p22_3)
(p22_1) edge[bend right=10]             (f)
(p22_2) edge           (p31_1)
(p22_3) edge[bend left=10]             (f)
(p22_4) edge[bend left=10]             (f)
(p31_1) edge[bend right=30]             (f)
(p32_1) edge[bend right=20]             (f)

 ;  
 
\path[->,dotted, thick]
(p11_1) edge           (p12_2) 
(p12_1) edge           (p21_2)
(p12_2) edge[bend left=10]             (f)
(p21_1) edge           (p22_1) 
(p21_2) edge           (p22_2)
(p21_3) edge           (p22_4)
(p22_1) edge           (p31_1) 
(p22_2) edge[bend right=35]             (f)
(p22_3) edge           (p31_1)
(p22_4) edge          (f)
(p31_1) edge           (p32_1)

(p32_1) edge[bend left=20]             (f)
 
  ;
  
 \node () at (2.5,3.5) {$\rightarrow$};  
 \node () at (3,4.2) {$\B^{\downarrow}_{17}$};
 \node (0)  at (3,3.5)  {$\F$};  
  
  \end{scope}

  \begin{scope}[shift={(8,-9)}]

  \end{scope}

\end{tikzpicture}

 \caption{ The OBDD refutation of  $\php_2$}
 \label{fig:php2}
\end{center}
\end{figure}
\end{example}
The following  resolution refutation simulates removing the nodes of the OBDD $\B_{17}$.

{ \small
 \begin{align*} 
  &p_{32}: & \lnot p_{12} \vee p_{31}  = \res(\lnot p_{12}\vee \lnot p_{32}, p_{31}\vee p_{32}) \\ 
   &   & \lnot p_{22} \vee p_{31}  = \res(\lnot p_{22}\vee \lnot p_{32}, p_{31}\vee p_{32})   
\end{align*}
}
{ \small
 \begin{align*} 
   &p_{31}:  & \lnot p_{11} \vee \lnot p_{12} =\res(\lnot p_{11}\vee \lnot p_{31}, \lnot p_{12} \vee p_{31} )  \\
 &    &\lnot p_{12} \vee \lnot p_{21} =\res(\lnot p_{21}\vee \lnot p_{31}, \lnot p_{12} \vee p_{31})   \\
 &     & \lnot p_{11} \vee \lnot p_{22} =\res(\lnot p_{11}\vee \lnot p_{31}, \lnot p_{22} \vee p_{31} )  \\
 &    &\lnot p_{21} \vee \lnot p_{22} =\res(\lnot p_{21}\vee \lnot p_{31}, \lnot p_{22} \vee p_{31})   
\end{align*}
}
{ \small
 \begin{align*} 
 &p_{22}: &  \lnot p_{11} \vee p_{21}    = \res(p_{21}\vee p_{22},  \lnot p_{11} \vee \lnot p_{22})    \\ 
 &   &  \lnot p_{12} \vee p_{21}    = \res(p_{21}\vee p_{22},  \lnot p_{12} \vee \lnot p_{22})     
\end{align*}
}
{ \small
 \begin{align*} 
& p_{21}: &  \lnot p_{11} = \res(\lnot p_{11}\vee \lnot p_{21},  \lnot p_{11} \vee p_{21})    \\
& & \lnot p_{12} = \res(\lnot p_{12}\vee \lnot p_{21},  \lnot p_{12} \vee p_{21})    
\end{align*}
}
{ \small
 \begin{align*} 
& p_{12}: &  p_{11} = \res(p_{11}\vee p_{12},  \lnot p_{12})     
\end{align*}
}
{ \small
 \begin{align*} 
& p_{11}: &  \bot  = \res(p_{11},  \lnot p_{11})   
\end{align*}
}

\section{The Main Result}\label{sec:main}

Now we establish   the main result   that any  OBDD refutation of an unsatisfiable CNF $\varphi$ of size $n$ can be simulated by 
a resolution refutation of $\varphi$ of size at most $O(n^2)$.

\begin{theorem}\label{theorem:main}
 Assume an unsatisfiable CNF $\varphi$. If there is an OBDD refutation of $\varphi$ of size $n$ then there is a resolution refutation of it of size 
 $O(|\cls(\varphi)|\cdot n)$.
\end{theorem}

\begin{proof}  
% Let  $$\B_1,\dots, \B_{2|\cls(\varphi)|-1}$$ be an OBDD refutation of $\varphi$ (as it is discussed in Section \ref{sec:obdd-proof-system}, the length of an OBDD refutation is $2|\cls(\varphi)|-1$), where the $|\cls(\varphi)|$ first OBDDs represent the clauses of $\varphi$. 

By Corollary  \ref{cor:ub}, elimination  of a node can be simulated by at most $|\cls(\varphi)|$ steps. Since the OBDD refutation has size $n$, we obtain that  there is a resolution refutation of $\varphi$ of size 
 $O(|\cls(\varphi)|\cdot n)$.

\end{proof}

Now,  Theorem \ref{theorem:main2} stating that  if there is an OBDD refutation of $\varphi$ of  size $n$ then  there is a resolution refutation of $\varphi$ of  size  $O(n^2)$  follows  straightforwardly from Theorem \ref{theorem:main}.
 
\begin{proof}(Proof of Theorem \ref{theorem:main2})  We can assume without loss of generality that $\varphi$ is a minimally unsatsfiable CNF, that is removing any clause from $\varphi$ will result in a satisfiable CNF.

Then the  size of any OBDD refutation of $\varphi$ is  at least $|\cls(\varphi)|$, that is $|\cls(\varphi)|\leq n$. By Theorem \ref{theorem:main},  there is a resolution refutation of $\varphi$ of size  $O(k)$, where
$$k\leq |\cls(\varphi)|\cdot n\leq n^2$$

\end{proof}

\section{Conclusions}\label{sec:conclusion}
The main reason for this study comes from the interest in providing theoretical explanations of the relative efficiency of algorithms used in SAT solving.

In this paper we show  that resolution simulates OBDDs polynomially if we limit both to CNFs and thus answered the open question of Groote and Zantema posed in \cite{GZ2003} whether there exists unsatisfiable  CNFs having polynomial OBDD refutations and exponentially long resolution refutations. 

The goal of this study was to show the existence of such a polynomial simulation rather than provide the tightest upper bound. We envisage that an OBDD refutation can be simulated by resolution refutation with  linear increase in size but it would require a more elaborated proof.  %A logical next step would be to answer the question whether OBDD $

\section*{Acknowledgments}
The author thanks Erika \'Abrah\'am   for helpful discussions at the early stage of this study.

\end{document}